\documentclass[]{imsart}

\doi{10.1214/154957804100000000}
\pubyear{0000}
\volume{0}

\RequirePackage{amsthm,amsmath,amsfonts,amssymb}
\RequirePackage[numbers]{natbib}
\RequirePackage[colorlinks,citecolor=blue,urlcolor=blue]{hyperref}
\RequirePackage{graphicx}
\RequirePackage{bbm}
\usepackage{placeins}

\startlocaldefs
\theoremstyle{plain}

\newtheorem{theorem}{Theorem}[section]
\newtheorem{lemma}[theorem]{Lemma}
\theoremstyle{remark}

\endlocaldefs

\endlocaldefs

\begin{document}

\begin{frontmatter}

\title{Conjugate Bayesian analysis of compound-symmetric Gaussian models}
\runtitle{Bayesian Compound Symmetric Gaussian Models}


\author{\fnms{Zachary M.} \snm{Pisano}\corref{}\ead[label=e1]{zpisano1@jhu.edu}}
\address{Department of Applied Mathematics and Statistics\\ Johns Hopkins University\\ \printead{e1}}

\runauthor{ZM Pisano}

\begin{abstract}
    We discuss Bayesian inference for a known-mean Gaussian population with a compound symmetric variance-covariance matrix. Since the space of such matrices is a linear subspace of that of positive definite matrices, we utilize the methods of \cite{Pisano-diss} to decompose the usual Wishart conjugate prior and derive a closed-form, bivariate conjugate prior distribution for the distinct entries of the compound-symmetric half-precision matrix. One samples from this density by transforming independent gamma random variables into the cone of possible entries of such matrices. We also demonstrate how the prior may be utilized to naturally test for the positivity of a common within-class correlation in a random-intercept model using two data-driven examples.
\end{abstract}

\begin{keyword}[class=MSC]
\kwd[Primary ]{62H05}
\kwd[; secondary ]{62F15}
\end{keyword}

\begin{keyword}
\kwd{Conjugate prior}
\kwd{exponential family}
\kwd{hierarchical model}
\kwd{random intercept model}
\end{keyword}



\end{frontmatter}

\section{Introduction}\label{sec:intro}
In parametric Bayesian inference conjugate prior distributions are comprised of equal parts practical utility and mathematical beauty: the former, in so far as such distributions' structure may greatly mitigate issues related to posterior computation and model selection; the latter, in so far as derivations demonstrating the conjugacy of individual priors may be satisfyingly and succinctly described via a simple rule by which the likelihood updates the hyperparameters \cite{Raiffa-1961}. Conjugate priors are guaranteed when the likelihood belongs to a full-rank exponential family in canonical form \cite{Diaconis-1979}, but a few non-exponential family examples (e.g., the Uniform-Pareto model, \cite{Fink1997ACO}, Sec. 2.5) have been catalogued as well. Due to the relatively limited number of likelihood settings which permit conjugate priors, Bayesian practitioners often implement priors which necessitate the use of Markov Chain Monte Carlo (MCMC) methods to approximate corresponding posteriors. While a few conjugate prior settings are ubiquitous and well-studied --- e.g., the Bayesian Gaussian model is the subject of an entire chapter in \cite{Hoff-2009} --- unless a likelihood possesses an obvious conjugate prior, MCMC methods are the norm, not the exception.

In this article we illustrate how defaulting to inherently computational methods --- thereby bypassing mathematical and analytical methods --- may result in a conjugate prior's being overlooked. In particular, we demonstrate how the distribution of a simple linear transformation of two independent gamma variables is conjugate for the unique entries a Gaussian model's precision matrix when this matrix is assumed to be compound symmetric. Throughout the article we assume
\begin{equation}\label{eq:model}
    \boldsymbol X_i\overset{i.i.d.}{\sim} \mathcal N_d\bigg(\boldsymbol\mu, \frac12 \boldsymbol{\mathcal H}^{-1}\bigg), \ \ i = 1,\dots,n
\end{equation}
where the $d\times d$ positive definite matrix $\boldsymbol{\mathcal H}$ denotes the half-precision, the canonical parameter for such a setting. Furthermore, we assume that $\boldsymbol{\mathcal H}$ is \emph{compound-symmetric}, i.e., structured such that
\begin{equation*}
    \boldsymbol{\mathcal H} = (\eta_1-\eta_2)\textbf I_d + \eta_2\boldsymbol 1_d \boldsymbol 1_d^\top
\end{equation*}
where $\eta_1 > 0, \ \eta_2 \in (-\frac{\eta_1}{d-1}, \eta_1)$, $\textbf{I}_d$ denotes the $d\times d$ identity matrix, and $\textbf 1_d$ denotes the $d$-vector of all ones \cite{Witkov}. Since it can be shown that the inverse of a compound-symmetric matrix is also compound symmetric (Appendix \ref{prelims}), a Gaussian model with compound-symmetric precision prescribes equal variance along each coordinate and equal correlation between each pair of coordinates; i.e., a Gaussian model with a compound-symmetric half-precision also possesses a compound-symmetric variance-covariance matrix
\begin{equation*}
    \boldsymbol\Sigma = (\sigma_1-\sigma_2)\textbf I_d + \sigma_2\boldsymbol 1_d \boldsymbol 1_d^\top
\end{equation*}
where $\sigma_1 > 0, \ \sigma_2 \in (-\frac{\sigma_1}{d-1}, \sigma_1)$.

Compound symmetry typically arises in linear models with repeated measurements, such as may be found in random intercept models within the study of medical treatments \cite{Demidenko-2004}, econometrics \cite{Maddala-1987}, and the effectiveness of educational methods \cite{Mulder-2019}. Bayesian inference for the compound-symmetric case has heretofore been conducted without closed-form conjugacy in mind. \cite{Spieg-2001} discussed a collection of priors for the intra-coordinate-correlation coefficient $\frac{\sigma_2}{\sigma_1}$, including a standard beta distribution on the unit interval which ignored the possibility that this parameter may take negative values. \cite{Mulder-2013} introduced a joint gamma-uniform prior which both permitted negative values of the intra-coordinate-covariance $\sigma_2$ and ensured posterior support on the space of positive definite compound symmetric matrices. More recently \cite{Mulder-2019} generalized the latter to a conditionally conjugate shifted/scaled beta distribution, and also discussed an appropriate prior for the noninformative case. Markov chain Monte Carlo methods were used both to sample from these priors and to approximate the posteriors; while useful in practice, such methods are inherently computationally complex.

Despite the development of such computationally intensive and sophisticated methods, it is nonetheless gratifying that we are able to fully express, in closed form, a conjugate prior for the precision matrix of a compound-symmetric Gaussian model. This prior possesses the mathematical convenience and beauty mentioned previously, as well as the benefits of being both interpretable from a data standpoint and robust to the addition of new data. Moreover, this prior also lends itself quite straightforwardly to the generation and simulation of compound symmetric matrices, as it may be derived as the distribution of a linear transformation of two independent gamma variates (Theorem \ref{thm:cs_prec_prior}). The same linear transformation applied to two independent inverse-gamma variates yields a conjugate prior for the entries of the variance-covariance matrix (Theorem \ref{thm:cs_var_prior}). All of this is explored in Section \ref{sec:primary_results}.

In Section \ref{sec:exp_fam} we explore an alternative means of deriving our primary result, which necessitates characterizing (\ref{eq:model}) as a fully rank-2 exponential family in canonical form, i.e., the general setting described in \cite{Diaconis-1979}. The model (\ref{eq:model}) can in fact be expressed as a linearly nested submodel of a Gaussian model with arbitrarily positive definite precision. These two facts in tandem lead us to invoke a result found in Chapter 3 of \cite{Pisano-diss}, which prescribes the structure of conjugate priors for such models. Here, we also derive the conditional prior of $\eta_2 \ | \ \eta_1$ and the marginal prior of $\eta_1$. The former is found to be a shifted/scaled Kummer-Beta distribution, whereas the latter is found to be a convolved-gamma distribution; in so doing we recover the exact form of the convolved-gamma density previously derived by \cite{DiSalvo-2006}.

Section \ref{sec:unknown_mean} extends the results of the previous sections to the unknown mean case. We prescribe a hierarchical prior in which the entries of the half-precision $\boldsymbol{\mathcal H}$ are \textit{a priori} distributed according to the conjugate prior for the known-mean case, and the mean vector given the half-precision has a multivariate Gaussian distribution \textit{\`a la} the usual normal-Wishart conjugate prior for general Gaussian model (Theorem \ref{thm:prec_prior_unknown_mean}). We prescribe an analogous conjugate hierarchical prior for the mean vector and the entries of the variance-covariance matrix $\boldsymbol\Sigma = \frac{1}{2}\boldsymbol{\mathcal H}^{-1}$ (Theorem \ref{thm:var_prior_unknown_mean}).

In Section \ref{sec:h_tests_for_cs} we demonstrate the priors' utility in testing for the postivity of the off-diagonal entry of $\boldsymbol\Sigma = \frac{1}{2}\boldsymbol{\mathcal H}^{-1}$. Such a test permits one to determine whether data arising from a compound symmetric setting may be more suitably modeled by the more restrictive random-intercept model \cite{Raudenbush-2002}. While this is most easily accomplished when the data consist entirely of vectors of equal length, the Expectation-Maximization (EM) algorithm \cite{DLR} permits us to extend the framework to data consisting of vectors of varying lengths. We conclude the body of the article with a discussion in Section \ref{sec:discussion}, and relegate to the appendices supplementary material detailing proofs of the various theorems (Appendices \ref{sec:proofs_of_theorems} and \ref{app:details}), the properties of compound symmetric matrices (Appendix \ref{prelims}) and the applicable EM algorithm (Appendix \ref{app:EM_for_RI}).


\section{Primary Results}\label{sec:primary_results}

Suppose we are in the setting (\ref{eq:model}) in which $\boldsymbol{\mathcal H} = (\eta_1-\eta_2)\textbf I_d + \eta_2\boldsymbol 1_d \boldsymbol 1_d^\top$, with the additional simplifying assumption that $\boldsymbol\mu = \textbf{0}$. Invocation of Lemma \ref{mat_det_lem} leads to us writing the compound-symmetric Gaussian likelihood as 
\begin{align*}
    L(\eta_1, \eta_2) = \pi^{-\frac{nd}{2}} &(\eta_1-\eta_2)^{\frac{n(d-1)}{2}}(\eta_1 + (d-1)\eta_2)^{\frac{n}{2}}\\
    &\times\exp\bigg\{-\eta_1\text{tr}(\boldsymbol s_n) - \eta_2\big(\boldsymbol1^\top \boldsymbol s_n\boldsymbol1-\text{tr}(\boldsymbol s_n)\big)\bigg\}
\end{align*}
where $\boldsymbol s_n = \sum_{i=1}^n \boldsymbol x_i\boldsymbol x_i^\top$. Based on this likelihood, a conjugate prior for $(\eta_1, \eta_2)$ up to a normalization constant should possess the form
\begin{equation}\label{eq:prior_form}
    (\eta_1-\eta_2)^{\theta_1}(\eta_1 + (d-1)\eta_2)^{\theta_2} \exp\{-\eta_1\theta_3-\eta_2\theta_4\}
\end{equation}
where $\theta_1,\dots,\theta_4$ denote hyperparameters. We observe that such a prior would indeed be conjugate, since the likelihood would update the hyperparameters such that
\begin{align*}
    \theta_1 &\mapsto \theta_1 + \frac{n(d-1)}{2}\\
    \theta_2 &\mapsto \theta_2 + \frac{n}{2}\\
    \theta_3 &\mapsto \theta_3 + \text{tr}(\boldsymbol s_n)\\
    \theta_4 &\mapsto \theta_4 + \big(\boldsymbol1^\top \boldsymbol s_n\boldsymbol1-\text{tr}(\boldsymbol s_n)\big).
\end{align*}
Moreover, this conjugate prior should have support on the two-dimensional space of compound-symmetric positive definite matrices which forms the cone 
\begin{equation*}
    \mathcal C_d := \bigg\{(\eta_1,\eta_2)\ : \ \eta_1>0, \frac{-\eta_1}{d-1}<\eta_2<\eta_1\bigg\}
\end{equation*}
when projected onto $\mathbb R$. This cone is depicted in Figure \ref{fig:cs_cone}.
\FloatBarrier
\begin{figure}
    \centering
    \includegraphics[width=.5\linewidth]{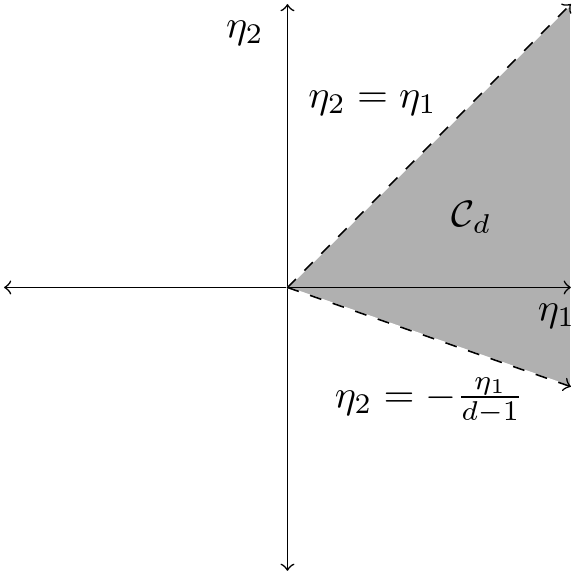}
    \caption{Two-dimensional representation of the space of compound-symmetric positive definite matrices.}
    \label{fig:cs_cone}
\end{figure}

\subsection{A Conjugate Prior for Compound-Symmetric Precision}
We observe that $\mathcal C_d$ is isomorphic to the first quadrant $\mathbb R^2_{>0}$. Hence, one possible means of deriving a proper prior for the compound-symmetric Gaussian model would be to first define a distribution supported on the first quadrant and then perform change-of-variables based around a mapping from $\mathbb R^2_{>0}$ to $\mathcal{C}_d$. Since the conjugate prior for the single precision parameter of a univariate Gaussian model is the gamma distribution \cite{Fink1997ACO}, we first define two independent (and not necessarily identically distributed) gamma variates $Y_j \sim \Gamma(\alpha_j,\lambda_j)\footnote{Throughout, we use the shape/rate formulation of gamma random variables.}, \ j = 1, 2$. As each $Y_j$ has strictly positive support, the pair $(Y_1, Y_2)$ is jointly supported on $\mathbb R^2_{>0}$. A simple transformation mapping elements of $\mathbb R^2_{>0}$ to elements of $\mathcal C_d$ is the linear transformation defined by
\begin{equation}\label{eq:lin_trans}
    \textbf{C}_d = \begin{bmatrix}
                        1 &1\\
                        1 &-\frac{1}{d-1}
                   \end{bmatrix}.
\end{equation}
That this linear transformation maps elements of $\mathbb R^2_{>0}$ into $\mathcal C_d$ can be seen by the fact that if $Y_1, Y_2 > 0$ we have
\begin{equation*}
    Y_1+Y_2>0 
\end{equation*}
and
\begin{equation*}
    -\frac{(Y_1+Y_2)}{d-1} < Y_1-\frac{Y_2}{d-1} < Y_1+Y_2.
\end{equation*}
Moreover the inverse transformation
\begin{equation*}
    \textbf{C}_d^{-1} = \begin{bmatrix}
        \frac{1}{d} &\frac{d-1}{d}\\ \frac{d-1}{d}& -\frac{d-1}{d}
        \end{bmatrix};
\end{equation*}
maps elements of $\mathcal C_d$ into $\mathbb R^2_{>0}$, since $\boldsymbol\eta\in \mathcal C_d$ implies
\begin{align*}
    0  &<\eta_1\\
    -\frac{\eta_1}{d-1} &< \eta_2 < \eta_1;
\end{align*}
the first inequality in the second line implies
\begin{equation*}
    0 < \frac{1}{d}\big(\eta_1 -(d-1)\eta_2\big),
\end{equation*}
and the second inequality implies
\begin{equation*}
    0<\bigg(\frac{d-1}{d}\bigg)(\eta_1-\eta_2)
\end{equation*}
as desired.

It so happens that all this gives us exactly what we want. We now arrive at our main result.
\begin{theorem}\label{thm:cs_prec_prior}
    If $Y_j \sim \Gamma(\alpha_j,\lambda_j),\ j = 1,2$, and $\textbf{C}_d = \begin{bmatrix}
                        1 &1\\
                        1 &-\frac{1}{d-1}
                   \end{bmatrix}$,
    then the distribution of $\textbf{C}_d\boldsymbol Y$ forms a conjugate prior for $\boldsymbol\eta$ in (\ref{eq:model}), with hyperparameter updates
    \begin{align*}
        \alpha_1 &\mapsto \alpha_1 +\frac{n(d-1)}{2}\\
        \alpha_2 &\mapsto \alpha_2 +\frac{n}{2}\\
        \lambda_1 &\mapsto \lambda_1 + \boldsymbol1^\top\boldsymbol s_n\boldsymbol1\\
        \lambda_2 &\mapsto \lambda_2 + \frac{d\text{tr}(\boldsymbol s_n)-\boldsymbol1^\top\boldsymbol s_n\boldsymbol1}{d-1}
    \end{align*}.
\end{theorem}

This result not only provides the exact functional form of the conjugate prior for $\boldsymbol\eta$ but also, by construction, provides a means by which this prior may be sampled. Therefore, any Bayesian hypothesis test of whether or not $\boldsymbol{\eta}$ lies within a specified region of $\mathcal C_d$ can be performed using elementary Monte Carlo methods conducted via built-in functions in statistical software packages (e.g., \textit{rgamma} in R).

Moreover, Theorem \ref{thm:cs_prec_prior} yields exact interpretability of the hyperparameters in terms of the data. The shape parameters $\alpha_1$ and $\alpha_2$ represent the volume of prior data corresponding to, respectively, estimates for the off-diagonal and diagonal entries of $\boldsymbol{ \mathcal H}$. Meanwhile, $\lambda_1$ and $\lambda_2$ convey information about the variance and covariance of prior data \textit{vis-\`a-vis} outer-products of data vectors. This interpretation is rendered clearer via the machinery of Section \ref{sec:exp_fam}.

\subsection{A Conjugate Prior for Compound-Symmetric Variance-Covariance}

Typically, the parameter of interest for Gaussian settings is the variance-covariance matrix $\boldsymbol\Sigma$ instead of the half-precision matrix $\boldsymbol{\mathcal H} = \frac{1}{2}\boldsymbol\Sigma^{-1}$, although the latter often takes greater precedence in the formulation of graphical models \cite{Lauritzen-1996}. Nonetheless, there is greater interest in Bayesian methods for the former; as mentioned in Section \ref{sec:intro}, Bayesian methods for compound-symmetric Gaussian models have thus far involved the construction of non-conjugate priors for the entries of $\boldsymbol\Sigma$ which either fail to ensure support on all of $\mathcal C_d$ or which require MCMC methods to perform posterior simulation or estimation \cite{Spiegelhalter-2002, Mulder-2013, Mulder-2019}.

The prior constructed in Theorem \ref{thm:cs_prec_prior} can be used as a stepping stone for that of compound-symmetric variance-covariance matrices. We can perform change-of-variables once again, using an invertible transformation motivated by Lemma \ref{cs_inverse_entries}, namely
\begin{align*}
    \sigma_1 &= \frac{\eta_1+(d-2)\eta_2}{2(\eta_1-\eta_2)(\eta_1+(d-1)\eta_2)}\\
    \sigma_2 &= -\frac{\eta_2}{2(\eta_1-\eta_2)(\eta_1+(d-1)\eta_2)}.
\end{align*}
In applying this change-of-variables to the conjugate prior for $\boldsymbol\eta$ obtained in Theorem \ref{thm:cs_prec_prior} we can obtain the exact analytical form of the conjugate prior for $\boldsymbol\sigma$. However, such a process might be more tedious than satisfying.

It turns out that we can derive a prior for $\boldsymbol\sigma$ as we did for $\boldsymbol\eta$ above. Writing the compound-symmetric Gaussian likelihood in terms of $\boldsymbol\sigma$ we have

\begin{align*}
    L(\sigma_1, \sigma_2) = (2\pi)^{-\frac{nd}{2}} &(\sigma_1-\sigma_2)^{-\frac{n(d-1)}{2}} (\sigma_1+(d-1)\sigma_2)^{-\frac{n}{2}}\\
    &\times \exp\bigg\{-\frac{1}{2} \bigg(\frac{\text{tr}(\boldsymbol s_n) - \frac{\boldsymbol{1}^\top\boldsymbol s_n \boldsymbol 1}{d}}{\sigma_1-\sigma_2} + \frac{\frac{\boldsymbol{1}^\top\boldsymbol s_n \boldsymbol 1}{d}}{\sigma_1 + (d-1)\sigma_2}\bigg)\bigg\}.
\end{align*}
The exponential term is obtained by expanding $\text{tr}(\boldsymbol\Sigma^{-1}\boldsymbol s_n)$ using Lemmata \ref{cs_trace_lemma} and \ref{cs_inverse_entries} and partial fraction decomposition. Based on this likelihood, a conjugate prior for $(\sigma_1, \sigma_2)$ up to a normalization constant should possess the form
\begin{equation}\label{eq:prior_form2}
    (\sigma_1-\sigma_2)^{-\theta_1} (\sigma_1+(d-1)\sigma_2)^{-\theta_2} \exp\bigg\{-\frac{\theta_3}{\sigma_1-\sigma_2} - \frac{\theta_4}{\sigma_1+(d-1)\sigma_2}\bigg\}
\end{equation}
where $\theta_1,\dots,\theta_4$ are hyperparameters as before. Such a prior would indeed be conjugate, since the likelihood would update the hyperparameters such that
\begin{align*}
    \theta_1 &\mapsto \theta_1+\frac{n(d-1)}{2}\\
    \theta_2 &\mapsto \theta_2+\frac{n}{2}\\
    \theta_3 &\mapsto \theta_3 + \frac{\text{tr}(\boldsymbol s_n) - \frac{\boldsymbol 1^\top \boldsymbol s_n\boldsymbol1}{d}}{2}\\
    \theta_4 &\mapsto \theta_4 + \frac{\boldsymbol 1^\top \boldsymbol s_n\boldsymbol1}{2d}.
\end{align*}
As before, this prior should have support on $\mathcal C_d$.

In the previous subsection we defined two independent gamma variates and linearly mapped them into $\mathcal C_d$ to obtain the conjugate prior for the compound-symmetric precision. A similar construction using independent inverse-gamma random variables gives us what we want.

\begin{theorem}\label{thm:cs_var_prior}
    Let $Z_j \sim \Gamma^{-1}(\alpha_j, \lambda_j), \ j = 1, 2$ be independent, and let $\textbf{C}_d$ be as in the previous theorem. The distribution of $\textbf{C}_d \boldsymbol Z$ forms a conjugate prior for $\boldsymbol\sigma$, with hyperparameter updates
    \begin{align*}
        \alpha_1 &\mapsto \alpha_1 + \frac{n}{2}\\
        \alpha_2 &\mapsto \alpha_2 + \frac{n(d-1)}{2}\\
        \lambda_1 &\mapsto \lambda_1 + \frac{\boldsymbol1^\top \boldsymbol s_n\boldsymbol 1}{2d^2}\\
        \lambda_2 &\mapsto \lambda_2 + \bigg(\frac{d-1}{2d^2}\bigg)\big(d\text{tr}(\boldsymbol s_n)-\boldsymbol1^\top \boldsymbol s_n\boldsymbol 1\big).
    \end{align*}
\end{theorem}

As in the previous subsection, this result not only yields an exact functional form of the conjugate prior for $\boldsymbol\sigma$ but also provides a means which this prior may be sampled, and interpretation of the hyperparameters and their updates in terms of the prior and new data. Any Bayesian hypothesis test of whether or not $\boldsymbol\sigma$ lies in a specified region of $\mathcal C_d$ can be performed by elementary Monte Carlo methods conducted via built-in statistical software packages (e.g., \textit{rinvgamma} in R). We give an example of such a hypothesis test in Section \ref{sec:h_tests_for_cs}.



\section{Primary Results: Exponential Family Formulation}\label{sec:exp_fam}

In the previous section we derived the conjugate prior for the half-precision of a compound symmetric Gaussian model using an appropriate transformation of independent gamma random variables. This is to say nothing as to how we determined which transformation, or which random variables to transform, or even whether a conjugate prior could be derived at all. Due to the partial fraction decomposition needed to derive the functional form (\ref{eq:prior_form2}) we suspect that previous work due to \cite{Mulder-2013,Mulder-2019, Spieg-2001} may have led to the belief that such a prior could not be derived; indeed, this difficulty arises only if we consider a prior for the entries of the variance-covariance matrix. However, if one instead re-characterizes such a search for a conjugate prior as one for the entries of the half-precision, one may turn to longstanding results for exponential families in canonical form.

\subsection{Conjugate Priors of Linearly Nested Canonical Exponential Families}\label{sec:conjugate_priors_for_lin}

Consider an i.i.d. sample of $n$ random $d$-vectors from a canonical $k$-rank exponential family for which, in the notation of \cite{Bic-Dok} we write the marginal for the $i$-th observation $x_i$ as $f(x_i| \boldsymbol\theta ) = h(x_i)\exp\{\langle\boldsymbol\theta,\boldsymbol T(x_i)\rangle - A(\boldsymbol\theta)\}$, with base measure $h(x_i)$, canonical parameter $\boldsymbol\theta \in \mathbbm R^k$, sufficient statistic $\boldsymbol T: \mathbbm R^d\to\mathbbm R^k$, and log-partition $A: \boldsymbol\Theta \to \mathbbm R$, where $\boldsymbol\Theta = \{\boldsymbol\theta\in \mathbbm R^k\  :\ |A(\boldsymbol\theta)|<\infty\}$. The likelihood of these $n$ observations is then
 \begin{equation*}
     L(\boldsymbol\theta) = \bigg(\prod_{i=1}^nh(x_i)\bigg)\exp\{\langle\boldsymbol\theta,\sum_{i=1}^n\boldsymbol T(x_i)\rangle - nA(\boldsymbol\theta)\}.
 \end{equation*}
 
 Per \cite{Diaconis-1979} the natural conjugate prior for $\boldsymbol\theta$ takes the form
 \begin{equation*}
     \rho(\boldsymbol\theta) = H(\boldsymbol\tau, m)\exp\{\langle\boldsymbol\theta,\boldsymbol\tau\rangle - mA(\boldsymbol\theta)\},
 \end{equation*}
 supported on $\boldsymbol\Theta$, where the normalization constant
 \begin{equation}
     H(\boldsymbol\tau, m)^{-1} = \int_{\boldsymbol\Theta}\exp\{\langle\boldsymbol\theta,\boldsymbol\tau\rangle - mA(\boldsymbol\theta)\} \partial\boldsymbol\theta.
 \end{equation}
 The hyperparameters possess the interpretation of a prior sample of size $m\in \mathbbm R_+$ which yielded the prior expectation of sufficient statistic $\boldsymbol\tau$ in the convex support of $m\boldsymbol T(X)$. The update rule which defines the posterior diestribution of $\boldsymbol\theta\ | \ X_{1:n}$ is
 \begin{align*}
     \boldsymbol\tau &\mapsto \sum_{i=1}^n \boldsymbol T(X_i) + \boldsymbol\tau\\
     m &\mapsto n+m
 \end{align*}
 
 We also define a linearly nested submodel of rank $\ell < k$, induced by a rank-$\ell$ matrix $\textbf{M}\in \mathbbm R^{\ell\times k}$ and canonical parameter space $\boldsymbol{\mathcal E} := \{\boldsymbol\eta \in \mathbbm R^\ell : |A(\textbf{M}^\top\boldsymbol\eta)|<\infty\}$, with density $g(x_i|\boldsymbol\eta) = h(x_i)\exp\{\langle\textbf{M}^\top\boldsymbol\eta,\boldsymbol T(x_i)\rangle - A(\textbf{M}^\top\boldsymbol\eta)\}$. The image $\textbf{M}^\top(\boldsymbol{\mathcal E})$ is a linear subset of $\boldsymbol\Theta$. In such case the nested model is a fully rank-$\ell$ exponential family with base measure $h$, canonical parameter $\boldsymbol\eta\in\boldsymbol{\mathcal E}$, natural sufficient statistic $\textbf{M}\boldsymbol T(X)$ and log-partition $B(\boldsymbol{\eta}) := A(\textbf{M}^\top\boldsymbol\eta)$ (\cite{Bic-Dok}, problem 1.6.17). Because of this one can invoke the methods of \cite{Diaconis-1979} to define a conjugate prior for $\boldsymbol\eta$ as
 \begin{equation*}
     \rho_N(\boldsymbol\eta) = G(\boldsymbol\upsilon, w)\exp\{\langle\boldsymbol\eta,\boldsymbol\upsilon\rangle - B(\boldsymbol\eta)\}
 \end{equation*}
 Here the hyperparameters $\boldsymbol\upsilon$ and $w$ act as $\boldsymbol\tau$ and $m$ above. Similarly we have
 \begin{equation}
     G(\boldsymbol\upsilon,w)^{-1} = \int_{\boldsymbol{\mathcal E}}\exp\{\langle\boldsymbol\eta,\boldsymbol\upsilon\rangle - B(\boldsymbol\eta)\}\partial\boldsymbol\eta
 \end{equation}
 with the hyperparameter update rule
 \begin{align*}
     \boldsymbol\upsilon &\mapsto \textbf{M}\bigg(\sum_{i=1}^n \boldsymbol T(X_i)\bigg) + \boldsymbol\upsilon\\
     w &\mapsto n+w.
 \end{align*}
In the following subsection we will demonstrate that the compound-symmetric Gaussian model is in fact a linear submodel of the general model.

\subsection{Compound Symmetry as a Linear Submodel}

Let us return to the model (\ref{eq:model}). To this end let our ``full'' model be that with $\boldsymbol{\mathcal H}$ arbitrarily positive definite, in which case the joint density of $\boldsymbol X_i, i =1,\dots,n$ as a canonical exponential family may be characterized by
\begin{align*}
    h(\boldsymbol x_i) &= \pi^{-\frac{d}{2}}\\
    \textbf T(\boldsymbol x_i) &= -\boldsymbol x_i\boldsymbol x_i^\top\\
    \boldsymbol\theta &=\boldsymbol{\mathcal H}\\
    \langle\boldsymbol\theta,\textbf T(\boldsymbol x_i)\rangle &= \text{tr}(\boldsymbol{\mathcal H}\textbf T(\boldsymbol x_i))\\
    A(\boldsymbol\theta) &= -\frac{1}{2}\log |\boldsymbol{\mathcal H}|.
\end{align*}
For the purpose of more easily expressing the compound symmetric setting as a linear submodel, we rewrite the canonical parameter and sufficient statistics as vectors in $\mathbbm R^{\frac{d(d+1)}{2}}$ and the inner product as a dot product thereupon as in Chapter 3 of \cite{Pisano-diss}; in so doing we obtain
\begin{align*}
    \boldsymbol\theta &= \begin{bmatrix} \text{diag}(\boldsymbol{\mathcal{H}}) &\triangle(\boldsymbol{\mathcal{H}})
    \end{bmatrix}^\top\\
    \textbf T(\boldsymbol x_i) &= -\begin{bmatrix}
    \text{diag}(\boldsymbol x_i\boldsymbol x_i^\top) &2\triangle(\boldsymbol x_i\boldsymbol x_i^\top)
    \end{bmatrix}^\top
\end{align*}
where diag($\cdot$) and $\triangle(\cdot)$ return as vectors the diagonal and upper-triangle of their square-matrix-valued arguments. It is clear that $\boldsymbol\theta^\top\textbf T(\boldsymbol x_i) = -\text{tr}(\boldsymbol{\mathcal H}\boldsymbol x_i \boldsymbol x_i^\top)$, and that $\boldsymbol\Theta$ is the set of all $\boldsymbol\theta \in \mathbbm R^{\frac{d(d+1)}{2}}$ corresponding to positive definite $\boldsymbol{\mathcal H}$. The conjugate prior for such $\boldsymbol{\mathcal H}$ is Wishart \cite{Fink1997ACO}:
\begin{equation*}
    \rho(\boldsymbol{\mathcal H}) = \frac{|\textbf{B}|^{\frac{m+d+1}{2}}}{\Gamma_d(\frac{m+d+1}{2})}|\boldsymbol{\mathcal H}|^{\frac{m}{2}}\exp\{-\text{tr}(\textbf{B}\boldsymbol{\mathcal H})\}, \ \boldsymbol{\mathcal H}\ \text{is positive definite}
\end{equation*}
with $m>0$ and $\textbf{B}\in \mathbbm R^{d\times d}$ positive definite. We have elected to present this density in shape-rate form to more easily relate it to the analogous parameterization for the gamma density it generalizes.

Observe that for every compound symmetric $\boldsymbol{\mathcal H}$ there exists $\boldsymbol\eta = (\eta_1,\eta_2)\in \mathcal C_d$ such that $\boldsymbol{\mathcal H} = (\eta_1-\eta_2)\textbf I + \eta_2\boldsymbol1\boldsymbol1^\top$, with the corresponding
\begin{equation*}
    \boldsymbol\theta = \begin{bmatrix}\eta_1\boldsymbol1_d^\top &\eta_2\boldsymbol1_{\frac{d(d-1)}{2}}^\top \end{bmatrix}^\top.
\end{equation*}
One notes that the linear transformation $\boldsymbol\theta = \textbf{M}^\top\boldsymbol\eta$ is satisfied by the rank-2 matrix
\begin{equation*}
    \textbf{M} = \begin{bmatrix} \boldsymbol1_d &\boldsymbol0_d\\
    \boldsymbol0_{\frac{d(d-1)}{2}} &\boldsymbol1_{\frac{d(d-1)}{2}}
    \end{bmatrix}.
\end{equation*}
I.e., the compound symmetric model can be expressed as a linear submodel of the arbitrary positive definite model. Therefore, the conjugate prior for the smaller model up to normalization is
\begin{equation}\label{eq:prec_prior_not_norm}
    \exp\bigg\{-\text{tr}\big(((\eta_1-\eta_2)\textbf I + \eta_2\boldsymbol1\boldsymbol1^\top)\textbf B\big) + \frac{m}{2}\log\big|(\eta_1-\eta_2)\textbf I + \eta_2\boldsymbol1\boldsymbol1^\top\big|\bigg\}.
\end{equation}
We see that $m$ plays the part of $w$, and that a function of $\textbf B$ plays the part of $\boldsymbol\upsilon$ in the previous subsection. In determining the normalization constant we shall see that we only care about the trace and sum of the off-diagonals of $\textbf B$ as far as hyperparameterization may be concerned.

Using Lemma \ref{mat_det_lem} and the properties of the matrix trace, we observe that (\ref{eq:prec_prior_not_norm}) may be rewritten as 
\begin{equation*}
      (\eta_1-\eta_2)^{\frac{m(d-1)}{2}}(\eta_1 + (d-1)\eta_2)^{\frac{m}{2}}\exp\{-\eta_1\beta_1 - \eta_2\beta_2\}
\end{equation*}
where $\beta_1 := \text{tr}(\textbf B)$ and $\beta_2 := \boldsymbol1^\top\textbf B\boldsymbol1-\text{tr}(\textbf B)$; Lemma \ref{trace_inequality} implies $-\beta_1<\beta_2<(d-1)\beta_1$ and, subsequently,
\begin{align*}
    0 &< \beta_1 +\beta_2\\
    0 &< (d-1)\beta_1-\beta_2.
\end{align*}

Theorem \ref{thm:cs_prec_prior} next tells us that the exact functional form of the prior is
\begin{equation*}
    \rho(\boldsymbol{\eta}) = Z_d(m,\beta_1,\beta_2)(\eta_1-\eta_2)^{\frac{m(d-1)}{2}} (\eta_1 + (d-1)\eta_2)^{\frac{m}{2}}\exp\{-\beta_1\eta_1-\beta_2\eta_2\}
\end{equation*}
with the normalization constant
\begin{equation}\label{eq:cs_prior_norm_cons}
    Z_d(m,\beta_1,\beta_2) := \frac{(\beta_1+\beta_2)^{\frac{m+2}{2}} ((d-1)\beta_1-\beta_2)^{\frac{m(d-1)+2}{2}}}{d^{\frac{md+2}{2}} \Gamma(\frac{m(d-1)+2}{2})\Gamma(\frac{m+2}{2})}.
\end{equation}
We can sample from $\rho(\boldsymbol\eta)$ by first sampling the independent variates $Y_1 \sim\Gamma(\frac{m+2}{2},\beta_1+\beta_2)$ and $Y_2 \sim \Gamma(\frac{m(d-1)+2}{2}, \beta_1-\frac{\beta_2}{d-1})$ on $\mathbb R^2_{>0}$ and transforming them into $\mathcal C_d$ via $\textbf{C}_d$.

The hyerparameter updates are
\begin{align*}
    \beta_1 &\mapsto \beta_1 + \text{tr}(\boldsymbol s_n)\\
    \beta_2 &\mapsto \beta_2 + \big(\boldsymbol1^\top \boldsymbol s_n\boldsymbol1-\text{tr}(\boldsymbol s_n)\big)\\
    m &\mapsto m+n.
\end{align*}
We see that $\beta_1$ conveys information about the diagonal entries of prior data's variance-covariance matrix, and $\beta_2$ may be interpreted in terms of the off-diagonal entries; both of these values are based on $m$ observations in our sample.

\subsection{Marginal and Conditional Priors for the Entries of the Half-Precision}\label{sec:marg_cond}

Theorem \ref{thm:cs_prec_prior} indeed gives us the constant (\ref{eq:cs_prior_norm_cons}), but we still have not addressed how we determined the relevance of the transformation $\textbf{C}_d$ and the gamma variates to be transformed. We outline our reasoning here, choosing to relegate the finer details to Appendix \ref{app:details}. If we did not have access to Theorem \ref{thm:cs_prec_prior}, one means of determining the normalizing constant of (\ref{eq:prec_prior_not_norm}) would be to directly compute
\begin{equation*}
    \int_{\mathcal C_d}(\eta_1-\eta_2)^{\frac{m(d-1)}{2}}(\eta_1 + (d-1)\eta_2)^{\frac{m}{2}}\exp\{-\eta_1\beta_1 - \eta_2\beta_2\}\partial\boldsymbol\eta.
\end{equation*}
Alternatively, we might instead attempt to write the non-normalized prior as the product of a non-normalized marginal prior for $\eta_1$ and a non-normalized conditional prior for $\eta_2\ | \ \eta_1$, either of which might be more easily identifiable.

The latter approach first reveals the conditional prior of $\eta_2\ | \ \eta_1$ as a Kummer-Beta distribution \cite{Ng-1995,Nagar-2002} shifted and scaled to a support of $(-\frac{\eta_1}{d-1}, \eta_1)$. Next, the marginal prior for $\eta_1$ is revealed to be a convolved-gamma distribution \cite{DiSalvo-2006}, i.e., that of the sum of independent $Y_1\sim\Gamma(\frac{m+2}{2},\beta_1+\beta_2)$ and $Y_2\sim\Gamma(\frac{m(d-1)+2}{2}, \beta_1-\frac{\beta_2}{d-1})$, which should be unsurprising in light of Theorem \ref{thm:cs_prec_prior}. With the exact form of the conjugate prior for $\boldsymbol\eta$ now in hand one can show that this quantity is exactly equal in distribution to $\textbf C_d\boldsymbol Y$ (Theorem \ref{thm:distribution}).

This effectively concludes our analysis when the model's mean is known. In the next section we extend our conjugate prior framework to include models whose mean is unknown.



\section{Conjugate Bayesian Analysis in the Unknown Mean Case}\label{sec:unknown_mean}
So far we have considered the setting (\ref{eq:model}) in which the true population mean is known. Suppose instead that
\begin{equation}\label{eq:model2}
    \boldsymbol X_{i} \overset{i.i.d.}{\sim}\mathcal N_d\bigg(\boldsymbol\mu,\frac{1}{2}\boldsymbol{\mathcal H}^{-1}\bigg), \ i = 1,\dots,n 
\end{equation}
in which $\boldsymbol\mu\in \mathbbm R^d$ is the unknown population mean and $\boldsymbol{\mathcal H}$ once again denotes the unknown $d\times d$ positive definite precision matrix. Without any further constraints on either of these parameters, the Normal-Wishart distribution (\cite{Hoff-2009}, Chapter 7) serves as a natural conjugate prior, with hierarchical structure
\begin{align*}
    \boldsymbol{\mathcal H} &\sim W_d(\alpha, \textbf B)\\
    \boldsymbol\mu\ | \ \boldsymbol{\mathcal H} &\sim \mathcal N_d\bigg(\boldsymbol\nu, \frac{1}{2\lambda}\boldsymbol{\mathcal H}^{-1}\bigg)
\end{align*}
where $\alpha > \frac{d+1}{2}$, $\textbf B\in \mathbbm R^{d\times d}$ is positive definite, $\lambda>0$ and $\boldsymbol\nu\in \mathbbm R^d$. As in Section \ref{sec:exp_fam}, the first two hyperparameters convey prior information about $\boldsymbol{\mathcal H}$; meanwhile $\boldsymbol\nu$ functions as the portion of $\boldsymbol\tau$ corresponding to a prior estimate of $\boldsymbol\mu$ from $\lambda$ observations. The conjugate prior for (\ref{eq:model2}) when $\boldsymbol{\mathcal H}$ is assumed to be conic compound-symmetric (i.e., supported on $\mathcal C_d$) is the subject of the following theorem.

\begin{theorem}\label{thm:prec_prior_unknown_mean}
    Suppose $m_{\boldsymbol\mu}, m_{\boldsymbol{\mathcal H}}>0, \boldsymbol\nu = (\nu_1,\dots,\nu_d)^\top\in \mathbbm R^{d}, \beta_1 > 0,$ and $\beta_2\in(-\beta_1,(d-1)\beta_1)$. Define independent gamma variates
    \begin{align*}
        Y_1 &\sim \Gamma\bigg(\frac{m_{\boldsymbol{\mathcal H}}+2}{2}, \beta_1 +\beta_2\bigg)\\
        Y_2 &\sim \Gamma\bigg(\frac{m_{\boldsymbol{\mathcal H}}(d-1)+2}{2},\beta_1 -\frac{\beta_2}{d-1}\bigg)
    \end{align*}
    and the linear transformation matrix $\textbf{C}_d$ as in (\ref{eq:lin_trans}).
    
    The hierarchical density of
    \begin{align*}
        \boldsymbol\eta &= \textbf{C}_d\boldsymbol Y\\
        \boldsymbol\mu\ |\ \boldsymbol{\mathcal H} = (\eta_1-\eta_2)\textbf I_d +\eta_2\boldsymbol 1_d\boldsymbol 1_d^\top &\sim \mathcal N_d \bigg(\boldsymbol\nu, \frac{1}{2m_{\boldsymbol\mu}}\boldsymbol{\mathcal H}^{-1}\bigg)
    \end{align*}
    forms a conjugate prior for (\ref{eq:model2}) with CS precision, with posterior hyperparameters defined by
    \begin{align*}
        m_{\boldsymbol\mu} &\mapsto m_{\boldsymbol\nu}+n\\
        m_{\boldsymbol{\mathcal H}} &\mapsto m_{\boldsymbol{\mathcal H}}+n\\
        \boldsymbol\nu &\mapsto \frac{m_{\boldsymbol\nu}\boldsymbol\nu+n\boldsymbol{\overline x}}{m_{\boldsymbol\nu}+n}\\
        \beta_1 &\mapsto \beta_1 + \text{tr}(\boldsymbol s_n)+\frac{m_{\boldsymbol\mu}n}{m_{\boldsymbol\mu}+n}\sum_{j=1}^d(\overline{x}_j-\nu_j)^2\\
        \beta_2 &\mapsto \beta_2 + \boldsymbol 1^\top \boldsymbol s_n\boldsymbol 1 - \text{tr}(\boldsymbol s_n) + \bigg(\frac{m_{\boldsymbol\mu}n}{m_{\boldsymbol\mu}+n}\bigg) \sum_{j\neq k}(\overline{x}_j-\nu_j)(\overline{x}_k-\nu_k),
    \end{align*}
    where $\boldsymbol{\overline x} = \lbrack\overline x_1 \cdots \overline x_d\rbrack^\top = \frac{\sum_{i=1}^n \boldsymbol x_i}{n}$ and $\boldsymbol{s}_n = \sum_{i=1}^n (\boldsymbol x_i -\boldsymbol{\overline x})(\boldsymbol x_i -\boldsymbol{\overline x})^\top$.
\end{theorem}

For a given collection of hyperparameters $m_{\boldsymbol\Sigma}, \lambda_1,\lambda_2, m_{\boldsymbol\mu},$ and $\boldsymbol\nu$ we can also define a conjugate prior for $(\boldsymbol\mu, \boldsymbol\sigma)$, and thereby extend Theorem \ref{thm:cs_var_prior} to the unknown mean case as well.

\begin{theorem}\label{thm:var_prior_unknown_mean}
    Suppose $m_{\boldsymbol\Sigma}, \lambda_1,\lambda_2, m_{\boldsymbol\mu}>0,$ and $\boldsymbol\nu = (\nu_1,\dots,\nu_d)^\top\in\mathbbm R^d$. Define independent inverse-gamma variates
    \begin{align*}
        Z_1 &\sim \Gamma^{-1}\bigg(\frac{m_{\boldsymbol{\Sigma}}+2}{2}, \lambda_1\bigg)\\
        Z_2 &\sim \Gamma^{-1}\bigg(\frac{m_{\boldsymbol{\Sigma}}(d-1)+2}{2},\lambda_2\bigg)
    \end{align*}
    and $\textbf C_d$ as above.
    
    The hierarchical density of
    \begin{align*}
        \boldsymbol\sigma &= \textbf{C}_d\boldsymbol Z\\
        \boldsymbol\mu \ | \ \boldsymbol\Sigma = (\sigma_1-\sigma_2)\textbf I_d +\sigma_2\boldsymbol1_d\boldsymbol1_d^\top &\sim \mathcal N_d\bigg(\boldsymbol\nu,\frac{1}{m_{\boldsymbol\mu}}\boldsymbol\Sigma\bigg)
    \end{align*}
    forms a conjugate prior for (\ref{eq:model2}) parameterized with a compound symmetric variance-covariance matrix, with posterior hyperparameters defined by
    \begin{align*}
        m_{\boldsymbol\Sigma} &\mapsto m_{\boldsymbol\Sigma} +n\\
        \lambda_1 &\mapsto\lambda_1 + \frac{\boldsymbol1^\top\boldsymbol s_n\boldsymbol1}{2d^2} +\frac{m_{\boldsymbol\mu}n}{2d^2(m_{\boldsymbol\mu}+n)}\big(\boldsymbol1^\top(\boldsymbol{\overline x}-\boldsymbol\nu)\big)^2\\
        \lambda_2 &\mapsto \lambda_2 + \frac{(d-1)(d\text{tr}(\boldsymbol s_n)-\boldsymbol1^\top\boldsymbol s_n\boldsymbol1)}{2d^2}+\bigg(\frac{(d-1)m_{\boldsymbol\mu}n}{m_{\boldsymbol\mu}+n}\bigg)\frac{d\|\boldsymbol{\overline{x}}-\boldsymbol\nu\|^2 - \big(\boldsymbol1^\top(\boldsymbol{\overline x}-\boldsymbol\nu)\big)^2}{2d^2}\\
        m_{\boldsymbol\mu} &\mapsto m_{\boldsymbol\mu}+n\\
        \boldsymbol\nu &\mapsto \frac{m_{\boldsymbol\mu}\boldsymbol\nu + n\boldsymbol{\bar x}}{m_{\boldsymbol\mu}+n}.
    \end{align*}
\end{theorem}

The proofs of Theorems \ref{thm:prec_prior_unknown_mean} and \ref{thm:var_prior_unknown_mean} have been relegated to Appendix \ref{sec:proofs_of_theorems}. The latter result permits us to straightforwardly conduct Bayesian hypothesis tests as to whether $\boldsymbol\sigma$ falls in a particular region of $\mathcal C_d$; we give an example of such a test in the following section.



\section{Test for the Positivity of a Common Within-Class Correlation in a Random-Intercept Model}\label{sec:h_tests_for_cs}

Suppose $D$ univariate observations $X_{ij}, \ j = 1,\dots,J,\ i = 1,\dots,d_j,\ \sum_{j=1}^J d_j = D$ arise from a random-intercept model (e.g., \cite{Raudenbush-2002}, Chapter 5):
\begin{align*}
    X_{ij} &= \mu + \mu_{0j}+\epsilon_{ij}\\
    \epsilon_{ij} &\overset{i.i.d.}{\sim} \mathcal N(0, \sigma_{\epsilon})\\
    \mu_{0j} &\overset{i.i.d.}{\sim} \mathcal N(0, \sigma_{\mu})
\end{align*}
where the $\mu_{0j}$ and $\epsilon_{ij}$ are independent of each other. In this two-level model each observation deviates from the global mean $\mu$ by a random quantity $\mu_{0j}$ particular to the group to which it belongs in addition to an error term $\epsilon_{ij}$ particular only to that observation. This model has been frequently implemented in educational research to model the nesting of children in classes or schools (e.g., \cite{Mulder-2013, Mulder-2019}). One can show, marginally, that the observations within a particular group share a common pairwise correlation $\frac{\sigma_\mu}{\sigma_\mu+\sigma_\epsilon}$. Expressing the model in the above conditional paradigm permits convenient partitioning of the total variation within and between classes.

However, when the primary research focus concerns the structure of within-group dependency and not necessarily how the total variation may be partitioned, one may treat the $\mu_{0j}$ as nuisance parameters and instead consider the marginal model
\begin{equation}\label{marginal_model}
    \boldsymbol X = \begin{bmatrix}
        \boldsymbol X_{\cdot 1}\\
        \vdots\\
        \boldsymbol X_{\cdot J}
    \end{bmatrix} \sim \mathcal N(\mu\boldsymbol 1_D, \boldsymbol\Sigma^*)
\end{equation}
in which
\begin{equation*}
    \boldsymbol\Sigma^* = \bigoplus_{j=1}^J \big(\sigma_\epsilon \textbf I_{d_j} + \sigma_{\mu}\boldsymbol 1_{d_j}\boldsymbol 1_{d_j}^\top\big),
\end{equation*}
where $\bigoplus$ denotes a direct sum. Note that the matrices on the block diagonal of $\boldsymbol\Sigma^*$ all possess compound symmetric structure. Under the conditional model it must necessarily be the case that $\sigma_\mu > 0$ so that the distribution of the $\mu_{0j}$ can be defined, but one can define a marginal model with $\sigma_\mu \in (-\frac{\sigma_\epsilon}{d_{\text{max}}}, 0\rbrack$. The marginal formulation therefore describes a greater class of models. In this section we are concerned with developing a hypothesis test to determine whether the more restrictive conditional formulation may be permitted; i.e.,
\begin{align*}
    H_0&:\ \sigma_\mu \leq 0\\
    \text{versus}\ H_a&:\ \sigma_\mu > 0.
\end{align*}
If we conclude $H_a$ we effectively determine that the data may be suitably modeled by a random-intercept model. Given the observed data vector $\boldsymbol x$ a na\"ive Bayesian hypothesis test consists of evaluating the posterior probability of $H_0$ versus that of $H_a$; if $P(H_0\ |\ \boldsymbol X = \boldsymbol x)>\frac{1}{2}$, then we conclude $H_0$. These hypotheses were considered by \cite{Mulder-2013}, who used MCMC methods to obtain posterior samples.

In line with our earlier work, the marginal model (\ref{marginal_model}) may be re-expressed in terms of its half-precision; if we rewrite $\boldsymbol\Sigma^*$ in terms of $\sigma_1= \sigma_\epsilon + \sigma_\mu$ and $\sigma_2 = \sigma_\mu$, we have that $ \boldsymbol\Sigma^* = \frac{1}{2}(\boldsymbol{\mathcal H}^*)^{-1}$ where
\begin{equation*}
    \boldsymbol{\mathcal H}^* = \frac{1}{2}\bigoplus_{j=1}^J\bigg((\eta_{1j}-\eta_{2j})\textbf I_{d_j} + \eta_{2j}\boldsymbol1_{d_j}\boldsymbol1^\top_{d_j}\bigg),
\end{equation*}
where the $\eta_{1j}$ and $\eta_{2j}$ are given by Lemma \ref{cs_inverse_entries}, i.e.,
\begin{align*}
    \eta_{1j} &= \frac{\sigma_1 +(d_j-2)\sigma_2}{(\sigma_1-\sigma_2)(\sigma_1+(d_j-1)\sigma_2))}\\
    \eta_{2j} &= \frac{-\sigma_2}{(\sigma_1-\sigma_2)(\sigma_1+(d_j-1)\sigma_2))}.
\end{align*}
When the $d_j$ are all equal, it is clear that the $\eta_{1j}$ and $\eta_{2j}$ are as well, and so the space of such $\boldsymbol{\mathcal H}^*$ forms a two-dimensional linear subspace of the $D\times D$ positive definite cone. In this case it is straightforward to sample from the posterior density of $\boldsymbol\sigma \ | \ \boldsymbol X=\boldsymbol x$ if we equip $\boldsymbol\sigma$ with an inverse-Gamma-type conjugate prior.

However, when the $d_j$ are not all equal the inverse-Wishart-type prior on $\mathcal C_{d_\text{max}}$ is not conjugate for $\boldsymbol\sigma$. In this case we propose a Gibbs sampling procedure \cite{gibbs-sampler} to sample from the posterior of $\boldsymbol\sigma\ | \ \boldsymbol X=\boldsymbol x$ based around the EM algorithm \cite{DLR}. The natural complete-data formulation of the marginal model (\ref{marginal_model}) consists of extending the vectors $\boldsymbol X_{\cdot j}$ each to a full length of $d_{\text{max}}$ with $d_{\text{max}}-d_j$ latent $Y$ random variates:
\begin{equation*}
    \begin{bmatrix}
        \boldsymbol X_{\cdot j}\\
        \boldsymbol Y_{\cdot j}
    \end{bmatrix}\ | \ \mu, \boldsymbol\sigma
    \overset{i.i.d.}\sim \mathcal N_{d_\text{max}}(\mu\boldsymbol1_{d_\text{max}}, (\sigma_1-\sigma_2)\textbf I_{d_\text{max}} +\sigma_2\boldsymbol1_{d_\text{max}}\boldsymbol1_{d_\text{max}}^\top).
\end{equation*}
Should we observe the latent random variables, the posterior for $(\mu,\boldsymbol\sigma)\ |  \ \boldsymbol X=\boldsymbol x, \boldsymbol Y = \boldsymbol y$ may be easily derived.
We also have that the latent variables, conditioned on $\boldsymbol X=\boldsymbol x,\mu$ and $\boldsymbol\sigma$, are normally distributed, i.e.,
\begin{equation*}
    \boldsymbol Y_{\cdot j}\ | \ \mu, \boldsymbol\sigma, \boldsymbol X= \boldsymbol x \sim \mathcal N_{d_{\text{max}}-d_j} (\boldsymbol\mu_j^*, \boldsymbol\Sigma_j^*)
\end{equation*}
where
\begin{align*}
    \boldsymbol\mu_{j}^* &= \mu\boldsymbol 1_{d_{\text{max}}-d_j} + \sigma_2\boldsymbol1_{d_{\text{max}}-d_j}\boldsymbol 1_{d_j}^\top \bigg\lbrack \frac{(\sigma_1 +(d_j-1)\sigma_2)\textbf I_{d_j}-\sigma_2\boldsymbol1_{d_j}\boldsymbol1_{d_j}^\top}{(\sigma_1-\sigma_2)(\sigma_1+(d_j-1)\sigma_2)}\bigg\rbrack(\boldsymbol x_{\cdot j} -\mu\boldsymbol 1_{d_j})\\
    \boldsymbol\Sigma_{j}^* &= (\sigma_1-\sigma_2)\textbf I_{d_{\text{max}}-d_j}+ \sigma_2\bigg(\frac{\sigma_1^2-2\sigma_1\sigma_2-\sigma_2^2(d_j-1)}{(\sigma_1-\sigma_2)(\sigma_1+(d_j-1)\sigma_2)}\bigg)\boldsymbol1_{d_{\text{max}}-d_j}\boldsymbol1_{d_{\text{max}}-d_j}^\top.
\end{align*}
The full conditional distributions of the $\boldsymbol Y_{\cdot, j},\mu,$ and $\boldsymbol\sigma$ are available to us, hence we propose the following Gibbs sampling procedure:
\begin{enumerate}
    \item Initialize $\boldsymbol\sigma = \boldsymbol\sigma^{(0)}, \mu = \mu^{(0)},$ and $\boldsymbol Y_{\cdot j} = \boldsymbol y_{\cdot j}^{(0)}$. Define $s = 0$.
    \item Set $(\mu,\boldsymbol\sigma) \mapsto (\mu^{(s+1)},\boldsymbol\sigma^{(s+1)})$ sampled from their posterior given $(\boldsymbol X, \boldsymbol Y)=(\boldsymbol x, \boldsymbol y^{(s)})$.
    \item Set $\boldsymbol Y_{\cdot j} \mapsto \boldsymbol y_{\cdot j}^{(s+1)}$ sampled from their normal density given $(\boldsymbol X, \mu, \boldsymbol\sigma) = (\boldsymbol x, \mu^{(s+1)}, \boldsymbol\sigma^{(s+1)})$. Set $s\mapsto s+1$.
    \item Repeat steps 2 and 3 until a desired sample size has been achieved.
\end{enumerate}

We propose two means of initializing the parameters and latent variables. One method is to first generate the parameters from their joint prior, and then the latent variables from their full conditional. Another is to carry out the EM algorithm induced by the complete data-formulation and initialize the parameters at their resulting MAP estimates, and then either initialize the latent random variables at their conditional expectations or at values randomly generated from their full conditionals. The former initialization does not require wrestling with the EM algorithm, but may require longer burn-in depending on the sample from the prior or may not even be possible if we use an uninformative or improper prior; whereas the latter may require less burn-in, but still requires an iterative scheme which must also be properly initialized. However, in computing the MAP estimates we uncover more information about the posterior; hence we recommend the EM algorithm, which we have relegated to Appendix \ref{app:EM_for_RI}.

Regardless of whether the $d_j$ are all equal or not, we can generate a sample from the posterior of $(\mu,\boldsymbol\sigma)$. The desired posterior probability of $H_0 \ |\ \boldsymbol X=\boldsymbol x$ can then be approximated as such:
\begin{align*}
    P(H_0\ | \ \boldsymbol X = \boldsymbol x) &= \int_{\sigma_2<0}\rho(\boldsymbol\sigma \ |\ \boldsymbol x)d\boldsymbol\sigma\\
    &\approx \frac{\sum_{s=1}^{S}\mathbbm 1(\sigma_2^{(s)}<0)}{S}
\end{align*}
where $S$ is the size of our posterior sample (excluding burn-in, if the $d_j$ are unequal). Again, if this approximation is greater than $\frac12$, we fail to reject $H_0$ in favor of the alternative that $\sigma_2$ is non-negative.

\subsection{Illustration}

We provide a simple illustration of this hypothesis test using two datasets found in Chapter 5 of \cite{Box-Tiao}; both datasets, printed in Tables \ref{tab:my-table}--\ref{tab:my-table2}, arise as $J=6$ samples all of size $d=5$. We assume a frequentist perspective and sample from the posterior in the absence of prior information, i.e., by equipping our model with the uninformative, improper prior
\begin{equation*}
    \rho(\mu,\boldsymbol\sigma) \propto (\sigma_1-\sigma_2)^{-(5-1)}(\sigma_1+(5-1)\sigma_2)^{-1},\  (\mu,\boldsymbol\sigma) \in \mathbbm R\times \mathcal C_5.
\end{equation*}
We estimated the desired posterior probability by sampling 100,000 times from the posterior of $\boldsymbol\eta$ and computing the proportion of samples with $\sigma_2 > 0$, which can be easily carried out in R using only primitive functions. We obtain
\begin{equation*}
    P(H_0\ | \ \boldsymbol X=\boldsymbol x) \approx 0.0080 < \frac{1}{2}
\end{equation*}
for the first dataset and
\begin{equation*}
    P(H_0\ | \ \boldsymbol X=\boldsymbol x) \approx 0.8246  > \frac{1}{2}
\end{equation*}
for the second. I.e., we conclude that the conditional formulation of the random intercept model is appropriate for the first dataset, but not so for the second; this was the conclusion reached by \cite{Mulder-2013}. Scatterplots of the posterior samples with density contours are depicted below.

\begin{table}[!ht]
\centering
\begin{tabular}{ccccccc}
Group           & 1    & 2    & 3    & 4    & 5    & 6    \\ \hline
                & 1545 & 1540 & 1595 & 1445 & 1595 & 1520 \\
                & 1440 & 1555 & 1550 & 1440 & 1630 & 1455 \\
                & 1440 & 1490 & 1605 & 1595 & 1515 & 1450 \\
                & 1520 & 1560 & 1510 & 1465 & 1635 & \textbf{1480} \\
                & 1580 & 1495 & 1560 & 1545 & \textbf{1625} & \textbf{1445} \\ \hline
$\bar{x}_j$ & 1505 & 1528 & 1564 & 1498 & 1600 & 1470 \\ \hline
$\bar x$    & \multicolumn{6}{c}{1527.5}             
\end{tabular}
\caption{The first of two datasets found in Chapter 5 of \cite{Box-Tiao}.}
\label{tab:my-table}
\end{table}

\begin{table}[!ht]
\centering
\begin{tabular}{ccccccc}
Group       & 1      & 2      & 3      & 4      & 5      & 6      \\ \hline
            & 7.298  & 5.220  & 0.110  & 2.212  & 0.282  & 1.722  \\
            & 3.846  & 6.556  & 10.386 & 4.852  & 9.014  & 4.782  \\
            & 2.434  & 0.608  & 13.434 & 7.092  & 4.458  & 8.106  \\
            & 9.566  & 11.788 & 5.510  & 9.288  & 9.446  & \textbf{0.758}  \\
            & 7.990  & -0.982 & 8.166  & 4.980  & \textbf{7.198}  & \textbf{3.758} \\ \hline
$\bar{x}_j$ & 6.2268 & 4.6380 & 7.5212 & 5.6848 & 6.0796 & 3.8252 \\ \hline
$\bar x$    & \multicolumn{6}{c}{5.6626}                         
\end{tabular}
\caption{The second of two datasets found in Chapter 5 of \cite{Box-Tiao}.}
\label{tab:my-table2}
\end{table}

\begin{figure}[!ht]
    \centering
    \includegraphics[width=\textwidth]{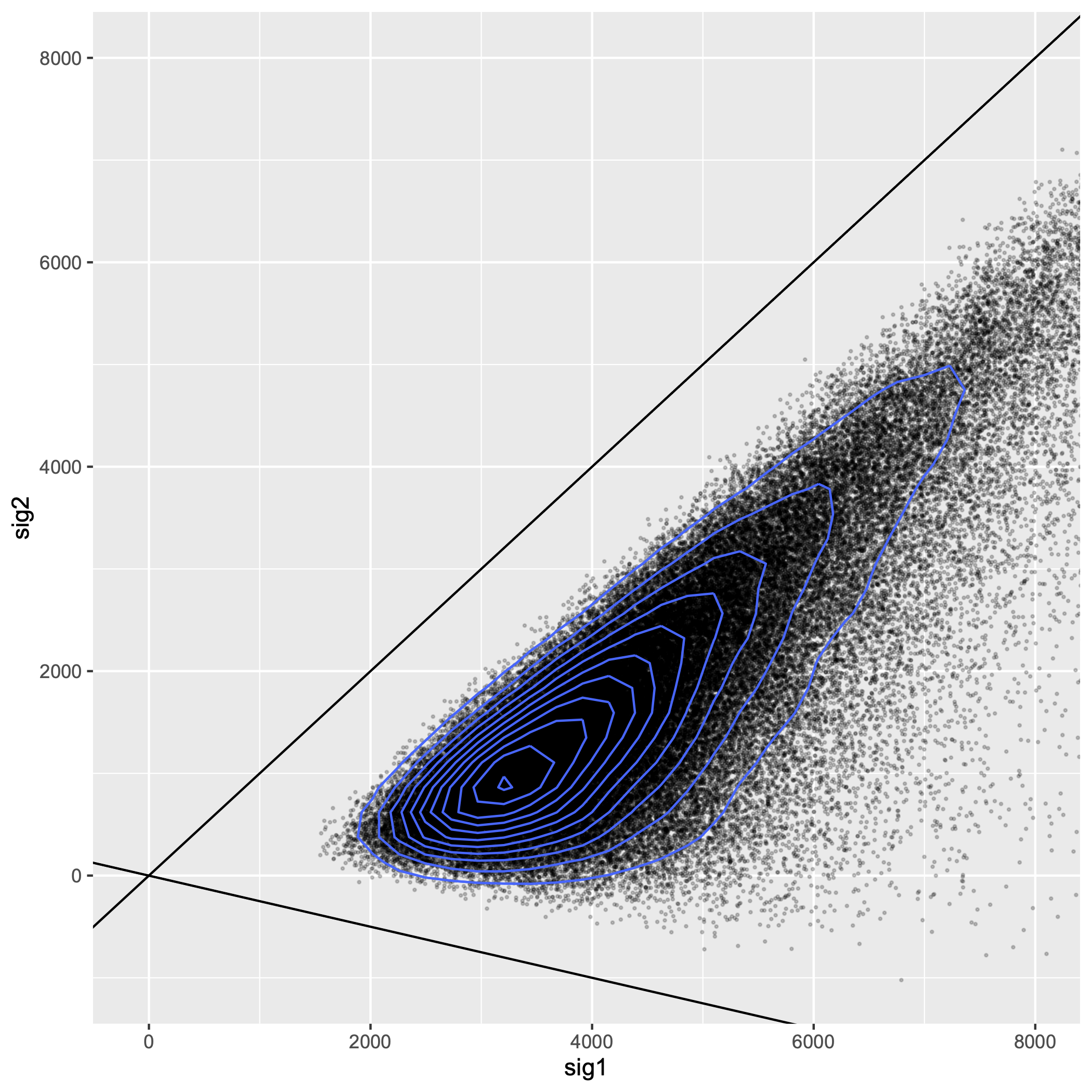}
    \caption{Posterior sample for the entries of the variance-covariance matrix for the dataset in Table \ref{tab:my-table}.}
    \label{fig:dataset1}
\end{figure}
\begin{figure}[!ht]
    \centering
    \includegraphics[width=\textwidth]{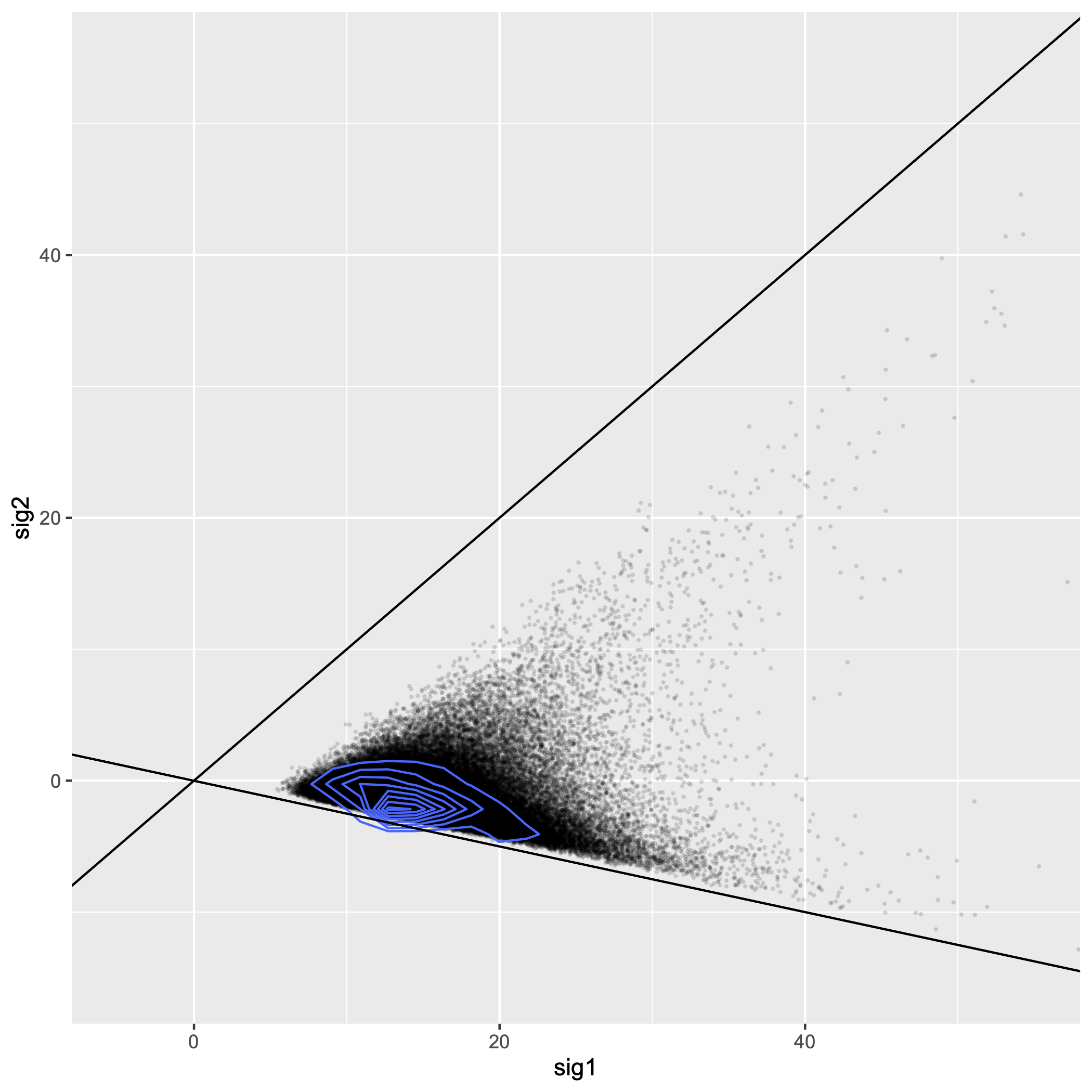}
    \caption{Posterior sample for the entries of the variance-covariance matrix for the dataset in Table \ref{tab:my-table2}.}
    \label{fig:dataset2}
\end{figure}

\FloatBarrier

We also repeated the test with the tables' emboldened entries removed so as to simulate settings in which the groups are of unequal size. For both modified datasets, we initialized the Gibbs sampler at the EM estimates of $\mu$ and $\boldsymbol\sigma$, burnt-in 1000 samples, and then iterated the sampler 100,000 times. We obtained
\begin{equation*}
    P(H_0\ | \ \boldsymbol X=\boldsymbol x) \approx 0.2632 < \frac{1}{2}
\end{equation*}
for the first dataset and
\begin{equation*}
    P(H_0\ | \ \boldsymbol X=\boldsymbol x) \approx 0.8664  > \frac{1}{2}
\end{equation*}
for the second, yielding the same conclusions as when we considered the entire datasets.

\FloatBarrier
\begin{figure}[h!]
    \centering
    \includegraphics[width=\textwidth]{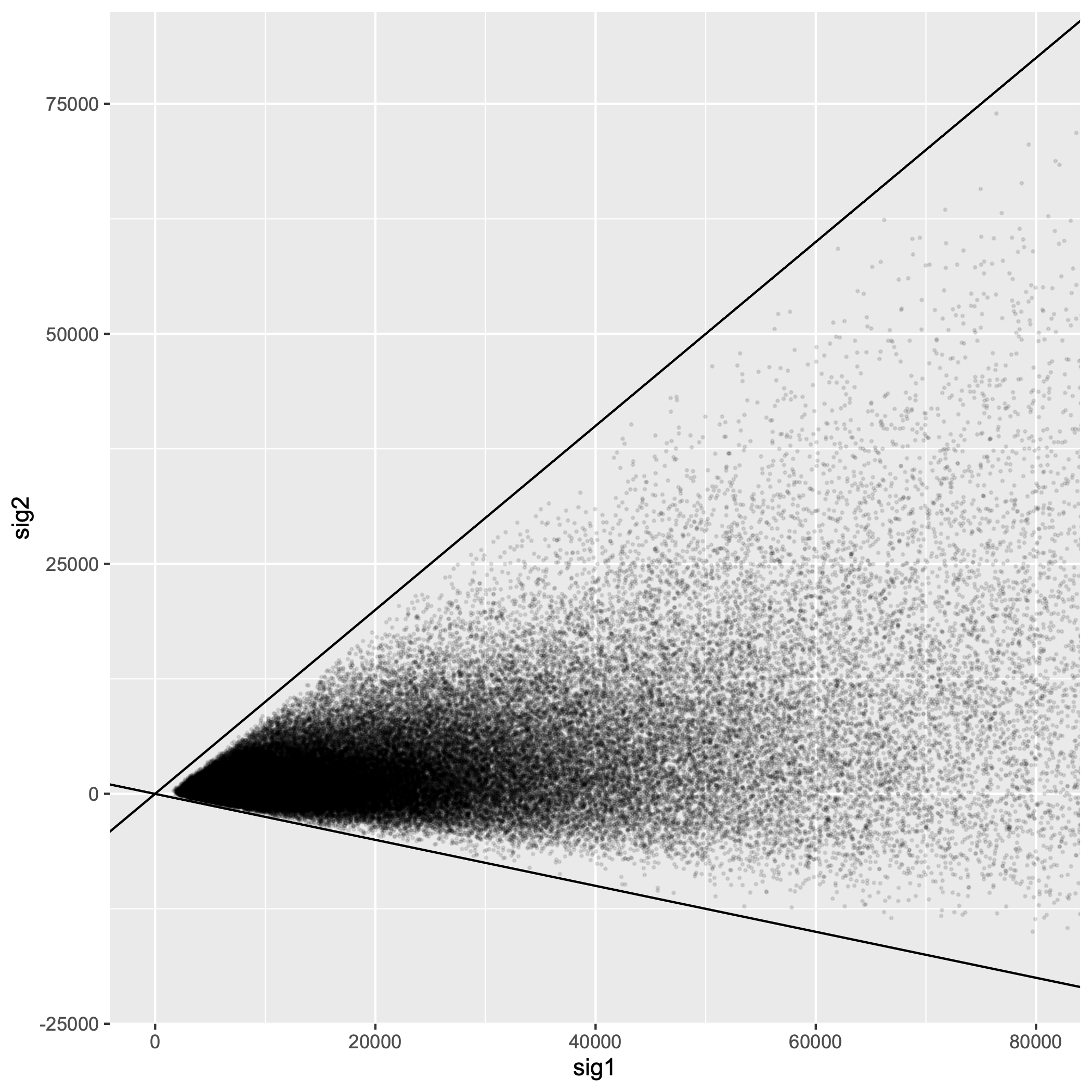}
    \caption{Posterior sample for the entries of the variance-covariance matrix for the dataset in Table \ref{tab:my-table} (emboldened data removed).}
    \label{fig:dataset1mod}
\end{figure}
\begin{figure}[h!]
    \centering
    \includegraphics[width=\textwidth]{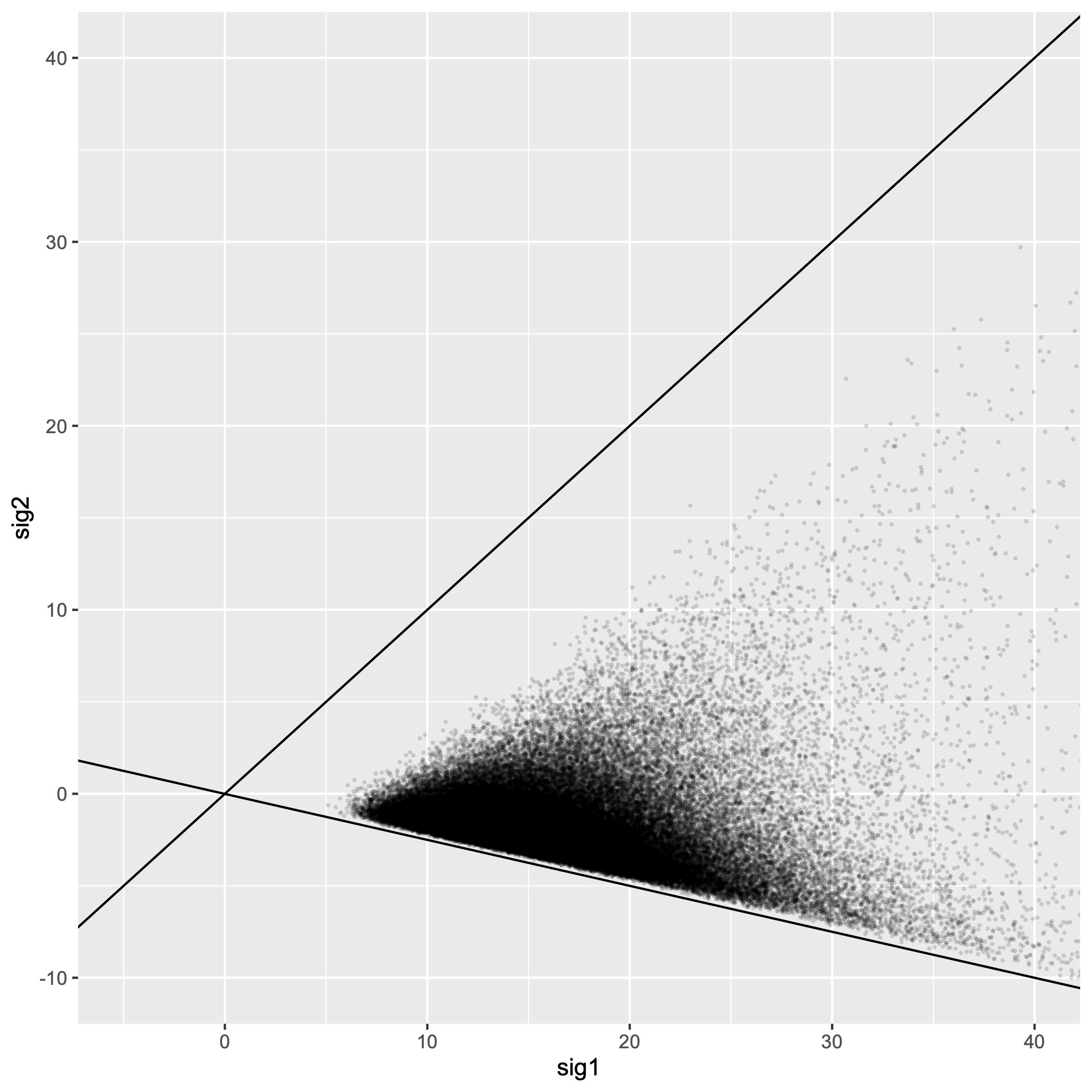}
    \caption{Posterior sample for the entries of the variance-covariance matrix for the dataset in Table \ref{tab:my-table2} (emboldened data removed).}
    \label{fig:dataset2mod}
\end{figure}
\FloatBarrier

\section{Discussion}\label{sec:discussion}

We have described conjugate prior distributions for two characterizations of a Gaussian model with a compound symmetric variance-covariance matrix. The derivation of these priors depended upon the expression of the desired model as a full-rank linear submodel of an exponential family in canonical form, a fact which allowed us to initially propose the existence of these priors before deriving them in terms of the canonical Wishart prior for the Gaussian half-precision. Although the exact densities which comprise the priors are nonstandard, simple changes-of-variables demonstrate that sampling from them is no more complicated than sampling from gamma or inverse-gamma distributions.

We suspect that similar approaches may yield conjugate priors for other linear submodels of the multivariate Gaussian model, e.g., Toeplitz \cite{cai-2013} and block-compound symmetric variance-covariance \cite{Coelho-2017}. In particular, we suspect that Wishart/inverse-Wishart generalizations of the priors described in Sections \ref{sec:primary_results}-\ref{sec:exp_fam} will arise as conjugate priors for the latter model. We plan to demonstrate this in future work.

The hypothesis test in Section \ref{sec:h_tests_for_cs} provides an example of the priors' utility beyond the mathematical convenience they exude. \cite{Mulder-2013,Mulder-2019, Fox-2017} considered related tests concerning the ordering and equality of the groups' respective interclass correlations. Our own conjugate priors are unsuitable for these tests, since the alternative models are not linear submodels of the canonical formulation of the multivariate Gaussian model.  If, however, the research question on hand concerns the ordering of interclass co-precisions then the Wishart-type conjugate prior on $\mathcal C_d$ may be extended to yield balanced Bayes factors across all possible orderings when the groups are of equal size.

The conjugate prior for the entries of a compound symmetric matrix can also be used to perform model selection to determine whether the variance-covariance matrix Gaussian data may be assumed to possess either arbitrarily positive definite structure or be restricted to a specific linear subset, e.g., diagonal, constant diagonal, or compound symmetric. Approximate model selection criteria such as BIC \cite{BIC} are not necessary for this task; since conjugate priors are available for all of these structures, each model's evidence may be computed in closed form. Moreover, the linearly nested relationship of these models amongst each other permits one to ``match'' their hyperparameters so as to mitigate the possibility of prior regularization being to used to arbitrarily select the final model \textit{a priori} (\cite{Pisano-diss}, Chapter 3; and \cite{Klugkist-2007}). The conjugate prior can also be used to carry out an exact Bayesian hypothesis test for repeated measures ANOVA \cite{lee-2015}; in this case we can simultaneously verify the sphericity assumption and carry out the ANOVA test itself by computing the evidences of the corresponding models.

Such considerations effectively proffer the priors we derived as the default Bayesian regularization for Gaussian models exhibiting compound symmetry. In all of these proposed applications, the usual decision criteria for statistical tests --- i.e., posterior probabilities of hypotheses and marginal likelihoods of the data --- can be obtained either in analytical form or after minimal Monte Carlo simulation. We contrast the implementation of these conjugate priors with those constructed in previous work, which only permit posterior inference after carrying out computationally intensive MCMC methods.



\bibliographystyle{imsart-number.bst} 
\bibliography{main.bib}       


\appendix
\section{Proofs of Theorems}\label{sec:proofs_of_theorems}

\begin{proof}[Proof of Theorem \ref{thm:cs_prec_prior}]\label{pf:cs_prec_prior}
    For the sake of notational recycling, let $\boldsymbol\eta = \textbf{C}_d\boldsymbol Y$. We shall derive the distribution of $\boldsymbol\eta$ via change-of-variables and verify that it possesses the form (\ref{eq:prior_form}). We have
    \begin{equation*}
        f_{\boldsymbol{Y}}(\boldsymbol y; \boldsymbol\alpha, \boldsymbol\lambda) = \frac{\lambda_1^{\alpha_1}\lambda_2^{\alpha_2}}{\Gamma(\alpha_1)\Gamma(\alpha_2)}y_1^{\alpha_1-1}y_2^{\alpha_2-1}\exp\{-\lambda_1 y_1-\lambda_2y_2\}
    \end{equation*}
    on $\mathbb R^2_{>0}$, and inverse transformation $\boldsymbol Y = \textbf{C}_d^{-1}\boldsymbol\eta$ in which
    \begin{equation*}
        \textbf{C}_d^{-1} = \begin{bmatrix}
            \frac{1}{d} &\frac{d-1}{d}\\ \frac{d-1}{d}& -\frac{d-1}{d}
            \end{bmatrix};
    \end{equation*}
    i.e.,
    \begin{align*}
        y_1 &= \frac{1}{d}\big(\eta_1 + \eta_2(d-1)\big)\\
        y_2 &= \bigg(\frac{d-1}{d}\bigg)(\eta_1\ - \eta_2).
    \end{align*}
    Hence, change-of-variables gives us 
    \begin{align*}
        f_{\boldsymbol\eta}(\boldsymbol\eta; \boldsymbol\alpha, \boldsymbol\lambda)  = \ &f_{\boldsymbol Y}(\textbf{C}_d^{-1}\boldsymbol\eta; \boldsymbol\alpha, \boldsymbol\lambda) \big||\textbf{C}_d|\big|^{-1}\\
        =\ &\big||\textbf{C}_d|\big|^{-1} \frac{\lambda_1^{\alpha_1}\lambda_2^{\alpha_2}}{\Gamma(\alpha_1)\Gamma(\alpha_2)}\\
        &\times\frac{(d-1)^{\alpha_2-1}}{d^{\alpha_1+\alpha_2-2}} (\eta_1-\eta_2)^{\alpha_2-1} \big(\eta_1 + \eta_2(d-1)\big)^{\alpha_1-1}\\
        &\times \exp\bigg\{-\frac{\lambda_1}{d}\big(\eta_1 + \eta_2(d-1)\big) - \lambda_2\bigg(\frac{d-1}{d}\bigg)(\eta_1\ - \eta_2)\bigg\}.
    \end{align*}
    Simplification yields
    \begin{align*}
        f_{\boldsymbol\eta}(\boldsymbol\eta; \boldsymbol\alpha, \boldsymbol\lambda) = \ &\frac{(d-1)^{\alpha_2}}{d^{\alpha_1+\alpha_2-1}}\times\frac{\lambda_1^{\alpha_1}\lambda_2^{\alpha_2}}{\Gamma(\alpha_1)\Gamma(\alpha_2)}(\eta_1-\eta_2)^{\alpha_2-1} \big(\eta_1 + \eta_2(d-1)\big)^{\alpha_1-1}\\
        &\times \exp\bigg\{-\bigg(\frac{\lambda_1+(d-1)\lambda_2)}{d}\bigg)\eta_1 -\bigg(\frac{d-1}{d}\bigg)(\lambda_1-\lambda_2)\eta_2\bigg\}
    \end{align*}
    which satisfies the functional form (\ref{eq:prior_form}). This is the general form of the conjugate prior, with
    \begin{align*}
        \theta_1 &= \alpha_2-1\\
        \theta_2 &= \alpha_1-1\\
        \theta_3 &= \frac{\lambda_1+(d-1)\lambda_2}{d}\\
        \theta_4 &= \bigg(\frac{d-1}{d}\bigg)(\lambda_1-\lambda_2).
    \end{align*}
\end{proof}

\begin{proof}[Proof of Theorem \ref{thm:cs_var_prior}]
    Our proof proceeds analogously to that of Theorem \ref{thm:cs_prec_prior}. For the sake of notational recycling, let $\boldsymbol\sigma = \textbf{C}_d\boldsymbol Z$. We shall derive the distribution of $\boldsymbol{\sigma}$ via change-of-variables and verify that it possesses the form (\ref{eq:prior_form2}). We have
    \begin{equation*}
        f_{\boldsymbol{Z}}(\boldsymbol z; \boldsymbol\alpha, \boldsymbol\lambda) = \frac{\lambda_1^{\alpha_1}\lambda_2^{\alpha_2}}{\Gamma(\alpha_1)\Gamma(\alpha_2)} z_1^{-\alpha_1-1} z_2^{-\alpha_2-1}\exp\bigg\{-\frac{\lambda_1}{z_1}-\frac{\lambda_2}{z_2}\bigg\}
    \end{equation*}
    on $\mathbb R^2_{>0}$ and inverse transformation $\boldsymbol Z = \textbf{C}_d^{-1}\boldsymbol\sigma$; i.e.,
    \begin{align*}
        z_1 &= \frac{1}{d}\big(\sigma_1 + (d-1)\sigma_2\big)\\
        z_2 &= \bigg(\frac{d-1}{d}\bigg)(\sigma_1-\sigma_2).
    \end{align*}
    Hence, change-of-variables gives us
    \begin{align*}
        f_{\boldsymbol\sigma}(\boldsymbol\sigma; \boldsymbol\alpha, \boldsymbol\lambda) = \ &f_{\boldsymbol Y}(\textbf{C}_d^{-1}\boldsymbol\eta; \boldsymbol\alpha, \boldsymbol\lambda) \big||\textbf{C}_d|\big|^{-1}\\
        =\ &\big||\textbf{C}_d|\big|^{-1} \frac{\lambda_1^{\alpha_1}\lambda_2^{\alpha_2}}{\Gamma(\alpha_1)\Gamma(\alpha_2)}\\
        &\times \frac{d^{\alpha_1+\alpha_2+2}}{(d-1)^{\alpha_2+1}}(\sigma_1-\sigma_2)^{-\alpha_2-1}\big(\sigma_1+(d-1)\sigma_2\big)^{-\alpha_1-1}\\
        &\times \exp\bigg\{-\frac{d\lambda_1}{\sigma_1+(d-1)\sigma_2} - \frac{\big(\frac{d}{d-1}\big)\lambda_2}{\sigma_1-\sigma_2}\bigg\}.
    \end{align*}
    Simplification yields
    \begin{align*}
        f_{\boldsymbol\sigma}(\boldsymbol\sigma; \boldsymbol\alpha, \boldsymbol\lambda) = \ &\frac{d^{\alpha_1+\alpha_2+1}}{(d-1)^{\alpha_2}} \times \frac{\lambda_1^{\alpha_1}\lambda_2^{\alpha_2}}{\Gamma(\alpha_1)\Gamma(\alpha_2)} (\sigma_1-\sigma_2)^{-\alpha_2-1}\big(\sigma_1+(d-1)\sigma_2\big)^{-\alpha_1-1}\\
        &\times \exp\bigg\{-\frac{d\lambda_1}{\sigma_1+(d-1)\sigma_2} - \frac{\big(\frac{d}{d-1}\big)\lambda_2}{\sigma_1-\sigma_2}\bigg\}.
    \end{align*}
    which satisfies the functional form (\ref{eq:prior_form2}). This is the general form of the conjugate prior, with
    \begin{align*}
        \theta_1 &= \alpha_2+1\\
        \theta_2 &= \alpha_1+1\\
        \theta_3 &= \bigg(\frac{d}{d-1}\bigg)\lambda_2\\
        \theta_4 &= d\lambda_1.
    \end{align*}
\end{proof}

\begin{proof}[Proof of Theorem \ref{thm:prec_prior_unknown_mean}]
    Define $\textbf B := (\frac{\beta_1}{d}-\frac{\beta_2}{d(d-1)})\textbf{I} + \frac{\beta_2}{d(d-1)}\textbf{11}^\top$. The logarithm of the proposed prior density for $(\boldsymbol\mu, \boldsymbol{\eta})$, after some massaging, can be written up to an additive constant only dependent upon the hyperparameters as
    \begin{equation}\label{nw_exp_term}
        -m_{\boldsymbol\mu}\text{tr}\big(\boldsymbol{\mathcal H}(\boldsymbol\mu-\boldsymbol\nu)(\boldsymbol\mu-\boldsymbol\nu)^\top\big) +\frac{1}{2}\log|2m_{\boldsymbol\mu}\boldsymbol{\mathcal H}| - \text{tr}(\boldsymbol{\mathcal H}\textbf{B}) + \frac{m_{\boldsymbol{\mathcal H}}}{2}\log|\boldsymbol{\mathcal H}|;
    \end{equation}
    we observe thus how the exponential term may be written similarly to that of the aforementioned normal-Wishart prior. Likewise the log-likelihood can be written similarly as 
    \begin{equation}\label{ll_exp_term}
        -\text{tr}\big(\boldsymbol{\mathcal H}\sum_{i=1}^n(\boldsymbol x_i-\boldsymbol\mu)(\boldsymbol x_i-\boldsymbol\mu)^\top\big) + \frac{n}{2}\log|\boldsymbol{\mathcal H}|.
    \end{equation}
    To derive the exponential term of the (non-normalized) posterior for $(\boldsymbol\mu,\boldsymbol\eta)$ we add (\ref{nw_exp_term}) and (\ref{ll_exp_term}) together since all of the omitted terms are constant with respect to both the data and parameters; due to the conjugacy of the normal-Wishart prior with the normal likelihood, this sum is exactly
    \begin{align*}
        &-(m_{\boldsymbol\mu}+n)\text{tr}\big(\boldsymbol{\mathcal H}(\boldsymbol\mu-\boldsymbol\nu_n)(\boldsymbol\mu-\boldsymbol\nu_n)^\top\big)\\ &+\frac{1}{2}\log|2(m_{\boldsymbol\mu}+n)\boldsymbol{\mathcal H}| - \text{tr}(\boldsymbol{\mathcal H}\textbf{B}_n) + \frac{m_{\boldsymbol{\mathcal H}}+n}{2}\log|\boldsymbol{\mathcal H}|
    \end{align*}
    where $\boldsymbol\nu_n$ is the updated value of $\boldsymbol\nu$ given in the statement of the theorem and
    \begin{equation*}
        \textbf{B}_n = \textbf{B} + \boldsymbol s_n + \bigg(\frac{m_{\boldsymbol\mu}n}{m_{\boldsymbol\mu}+n}\bigg) (\boldsymbol{\overline{x}} - \boldsymbol\nu)(\boldsymbol{\overline{x}} - \boldsymbol\nu)^\top.
    \end{equation*}
    Now, using the fact that $\boldsymbol{\mathcal H}$ is compound symmetric, we have via Lemma \ref{cs_trace_lemma} that $\text{tr}(\boldsymbol{\mathcal H} \textbf{B}_n) = \text{tr}(\boldsymbol{\mathcal H}\textbf{B}^\prime_n)$, where $\textbf{B}^\prime_n$ is the nearest compound symmetric matrix to $\textbf{B}_n$ obtained via Lemma \ref{optimal_CS_mat}, easily obtained as $(\frac{\beta_1^{(n)}}{d}-\frac{\beta_2^{(n)}}{d(d-1)})\textbf{I} + \frac{\beta_2^{(n)}}{d(d-1)}\textbf{11}^\top$ where $\beta_1^{(n)}$ and $\beta_2^{(n)}$ are the updated hyperparameters given in the statement of the theorem.
    
    Thus the sum of (\ref{nw_exp_term}) and (\ref{ll_exp_term}) is exactly the exponential term of the joint density of $(\boldsymbol\mu,\boldsymbol\eta)$ parameterized with the updated hyperparameters.   
\end{proof}

\begin{proof}[Proof of Theorem \ref{thm:var_prior_unknown_mean}]
    If we define $Y_1 = 2d^2 Z_1^{-1}$ and $Y_2 = \frac{2d^2}{(d-1)^2}Z_2^{-1}$, we are precisely in the setting of the previous theorem, with
    \begin{align*}
        m_{\boldsymbol{\mathcal{H}}} & = m_{\boldsymbol{\Sigma}}\\
        \beta_1 &= \frac{\big(2d^3-(d-1)\big)\lambda_1+(d-1)^3\lambda_2}{2d^3}\\
        \beta_2 &= \frac{(d-1)\lambda_1-(d-1)^3\lambda_2}{2d^3}.        
    \end{align*}
    The new hyperparameters are updated as in the previous theorem which, after re-expressing them in terms of $m_{\boldsymbol{\Sigma}}, \lambda_1, \text{ and } \lambda_2$ yields the updates in the theorem statement.
\end{proof}

\section{Positive Definite and Compound Symmetric Matrices}\label{prelims}

It is known that the set of $d\times d$ positive semidefinite matrices forms a cone in $\mathbbm R^{\frac{d(d+1)}{2}}$, where the dimension is the number of potentially unique entries of such matrices \cite{Hill-1987}. If one considers only the interior of this cone (itself a convex cone), we instead have the space of $d\times d$ positive definite matrices. In this article we are largely concerned with the two-dimensional subset of the latter cone associated with compound symmetric matrices; we call this the \textit{($d$-th) compound symmetric cone} and denote its projection onto $\mathbbm R^2$ as
\begin{equation*}
    \mathcal C_d := \{(x,y)\ : \ x>0, \frac{-x}{d-1}<y<x\}.
\end{equation*}
One immediately observes that $(x,y)\in \mathcal C_d$ iff the matrix $(x-y)\textbf I_d +y\boldsymbol 1_d\boldsymbol 1_d^\top$ is compound symmetric. We also consider an open two-dimensional rectangle (which we call the \textit{($d$-th) compound symmetric rectangle}) isomorphic to $\mathcal C_d$:
\begin{equation*}
    \mathcal R_d := \{(x,y) \ : \ x>0, \frac{-1}{d-1}< y < 1\}.
\end{equation*}
We also see that $(x,y)\in \mathcal R_d$ iff $x\big((1-y)\textbf I_d + y\boldsymbol 1_d\boldsymbol 1_d^\top\big)$ is compound symmetric. That $\mathcal R_d $ and $\mathcal C_d$ are isomorphic is obvious, since one can easily define $(x,y)\mapsto(x,\frac{y}{x})$, with inverse $(x,y)\mapsto(x,xy)$, which is a bijection between the two sets.

Key to our derivation of the marginal and conditional priors of $\eta_1$ and $\eta_2\ | \ \eta_1$ is the following lemma, which bounds the average off-diagonal entry of a positive definite matrix in terms of the average diagonal entry.
\begin{lemma}\label{trace_inequality}
    If $\textbf B \in \mathbbm R^{d\times d}$ ($d\geq2$) is positive definite, then 
    \begin{equation*}
        -\frac{\text{tr}(\textbf B)}{d(d-1)} < \frac{\boldsymbol 1^\top \textbf B \boldsymbol 1 - \text{tr}(\textbf B)}{d(d-1)} < \frac{\text{tr}(\textbf B)}{d}.
    \end{equation*}
\end{lemma}

\begin{proof}
    From the definition of positive definiteness, we have that $0 < \boldsymbol 1^\top \textbf B\boldsymbol 1$. Subtracting $\text{tr}(\textbf B)$ from both sides and dividing through by $d(d-1)$ yields the lower bound.

    To obtain the upper bound, we use the Spectral Theorem to write $\textbf B = \textbf{U}_{\textbf B} \boldsymbol{\Lambda}_{\textbf B} \textbf{U}_{\textbf B}^\top$ in which the rows of $\textbf U_{\textbf B}$ are unit eigenvectors of $\textbf B$ and $\boldsymbol\Lambda_{\textbf{B}} = \text{diag}(\lambda_1, \dots, \lambda_d)$ is a diagonal matrix containing the eigenvalues of $\textbf{B}$, subsequently giving us
    \begin{align*}
       \frac{\boldsymbol 1^\top\textbf B\boldsymbol 1}{d} &= \bigg(\frac{\boldsymbol 1}{\sqrt d}\bigg)^\top\textbf{U}_{\textbf B} \boldsymbol{\Lambda}_{\textbf B} \textbf{U}_{\textbf B}^\top \bigg(\frac{\boldsymbol 1}{\sqrt d}\bigg)\\
        &= \bigg(\textbf{U}_{\textbf B}^\top\frac{\boldsymbol 1}{\sqrt d}\bigg)^\top\boldsymbol{\Lambda}_{\textbf B}\bigg(\textbf{U}_{\textbf B}^\top\frac{\boldsymbol 1}{\sqrt d}\bigg) \\
        &= \sum_{j=1}^d\lambda_{j}\bigg(\boldsymbol u_j^\top\frac{\boldsymbol1}{\sqrt d}\bigg)^2.
    \end{align*}
    Note that $\frac{\boldsymbol 1}{\sqrt d}$ and the $\boldsymbol u_j$ are all unit vectors, hence the squared dot product in each summand is 
    \begin{equation}\label{cos_ineq}
        \text{cos}^2(\theta(\boldsymbol u_j, \frac{\boldsymbol 1}{\sqrt d}))\leq 1,
    \end{equation}
    where $\theta(\boldsymbol u, \boldsymbol v)$ is the measure of the angle between vectors $\boldsymbol u$ and $\boldsymbol v$, with equality only occurring when $\boldsymbol u_j^\top\frac{\boldsymbol1}{\sqrt d} = \pm 1$. On the unit sphere in $\mathbbm R^d$ this only occurs when $\boldsymbol u_j = \pm \frac{\boldsymbol 1}{\sqrt d}$, and if this is true for two or more of the $\boldsymbol u_j$ then those eigenvectors are not linearly independent, which would contradict our first assumption that $\textbf B$ be positive definite, i.e., necessarily full-rank with linearly independent eigenvectors. Therefore the inequality (\ref{cos_ineq}) is strict for at least one of the $\boldsymbol u_j$.

    We now have
    \begin{align*}
        \text{tr}(\textbf B) &= \sum_{j=1}^d\lambda_j\\
        &> \sum_{j=1}^d\lambda_j \bigg(\boldsymbol u_j^\top\frac{\boldsymbol1}{\sqrt d}\bigg)^2\\
        &= \frac{\boldsymbol 1^\top\textbf B\boldsymbol 1}{d}.
    \end{align*}
    Subtracting $\frac{\text{tr}(\boldsymbol B)}{d}$ from both sides and dividing by $d-1$ yields the desired upper bound.
\end{proof}

Two immediate consequences follow. The first is that
\begin{equation*}
    -\text{tr}(\textbf B) < \boldsymbol 1^\top \textbf B \boldsymbol 1 - \text{tr}(\textbf B) < (d-1)\text{tr}(\textbf B).
\end{equation*}
The second is that a compound symmetric matrix may be naturally constructed from any positive definite matrix. When we construct the conjugate prior for the compound-symmetric half-precision on $\mathcal C_d$ in the following section as a linear submodel of the unconstrained (i.e., arbitrarily positive definite) case, this result permits us to begin with any positive definite rate matrix $\textbf B$ in the original Wishart prior; said matrix need not be compound symmetric itself.

\begin{lemma}\label{optimal_CS_mat}
    Suppose $\textbf B \in \mathbbm R^{d\times d}$ ($d\geq 2$) is positive definite, and define $\hat b_1 = \frac{\text{tr}(\textbf B)}{d}$ and $\hat b_2 = \frac{\boldsymbol 1^\top \textbf B \boldsymbol 1 - \text{tr}(\textbf B)}{d(d-1)}$. The matrix
    \begin{equation*}
        \textbf{\^B} := (\hat b_1-\hat b_2)\textbf I_d + \hat b_2\boldsymbol 1_d\boldsymbol 1_d^\top
    \end{equation*}
    solves the optimization problem
    \begin{align*}
        \underset{\textbf B'}{\text{min}} \ &\|\textbf B-\textbf B'\|^2_F\\
        \text{s.t.} \ & \textbf B' \ \text{is compound symmetric}
    \end{align*}
    where $\|\cdot\|_F$ denotes the usual Frobenius norm.
\end{lemma}

\begin{proof}
    We consider a relaxed version of the optimization problem and show that its solution is in fact compound symmetric; let us solve instead
    \begin{align*}
        \underset{b_1,b_2}{\text{min}} \ &\|\textbf B-(b_1-b_2)\textbf I_d - b_2\boldsymbol1_d\boldsymbol1_d^\top\|_F^2\\
        \text{s.t.}\ &b_1,b_2\in \mathbbm R.
    \end{align*}
    Let $\mathcal O(b_1,b_2)$ denote the objective function; from the definition of the Frobenius norm, we have
    \begin{equation*}
        \mathcal O(b_1,b_2) = \sum_{i=1}^d (b_{ii}-b_1)^2 + \sum_{i\neq j} (b_{ij}-b_2)^2,
    \end{equation*}
    which is clearly convex in $(b_1,b_2)$, with partial derivatives proportional to
    \begin{align*}
        \frac{\partial\mathcal O(b_1,b_2)}{\partial b_1} &\propto \text{tr}(\textbf B) - db_1\\
        \frac{\partial\mathcal O(b_1,b_2)}{\partial b_2} &\propto \boldsymbol 1_d^\top \textbf B\boldsymbol 1_d - \text{tr}(\textbf B) - d(d-1)b_2.
    \end{align*}
    Setting these equal to 0 and solving for $b_1$ and $b_2$ gives the $\hat b_1$ and $\hat b_2$ in the statement of the lemma. Lemma \ref{trace_inequality} yields the fact that $(\hat b_1,\hat b_2)\in \mathcal C_d$, which is true iff the above $\textbf{\^B}$ is compound symmetric.
\end{proof}
Heuristically, the nearest compound-symmetric matrix to a positive definite matrix (in Frobenius norm) is that defined by the latter's average diagonal entry and average off-diagonal entry.

The following result assists our derivation of the conjugate prior for the parameters of (\ref{eq:model}) when both $\boldsymbol\mu$ and $\boldsymbol{\mathcal{H}}$ are unknown.
\begin{lemma}\label{cs_trace_lemma}
    If $\textbf{B},\textbf{\^B}$ are as in Lemma \ref{optimal_CS_mat} and $\textbf{A}$ is $d\times d$ compound symmetric, then $\text{tr}(\textbf{AB})= \text{tr}(\textbf{A\^B})$.
\end{lemma}
\begin{proof}
    Let $(a_1,a_2)\in \mathcal C_d$ denote the entries of $\textbf{A}$. Note that $\text{tr}(\textbf{B}) = \text{tr}(\textbf{\^B})$ and $\text{tr}(\textbf{B11}^\top) = \boldsymbol 1^\top\textbf{B}\boldsymbol 1 = \boldsymbol 1^\top\textbf{\^B}\boldsymbol 1= \text{tr}(\textbf{\^B11}^\top)$. We have
    \begin{align*}
        \text{tr}(\textbf{AB}) &= \text{tr}\bigg(\big((a_1-a_2)\textbf{I}+a_2\boldsymbol{11}^\top\big)\textbf{B}\bigg)\\
        &= (a_1-a_2)\text{tr}(\textbf{B}) + a_2\text{tr}(\textbf{B}\boldsymbol{11}^\top)\\
        &= (a_1-a_2)\text{tr}(\textbf{\^B}) + a_2\text{tr}(\textbf{\^B}\boldsymbol{11}^\top)\\
        &= \text{tr}(\textbf{A\^B})
    \end{align*}
    as desired.
\end{proof}

We also invoke the Matrix Determinant Lemma (ex. 1.3.24 of \cite{Horn-Johnson}) to obtain the determinant of a positive definite compound symmetric matrix.

\begin{lemma}\label{mat_det_lem}
    If $\textbf B = (b_1-b_2)\textbf I + b_2\boldsymbol 1\boldsymbol 1^\top$ is $d\times d$ positive definite compound symmetric, then
    \begin{equation*}
        |\textbf B| = (b_1-b_2)^{d-1}(b_1+(d-1)b_2)
    \end{equation*}
\end{lemma}

Finally, since we are interested in inference for a compound symmetric variance covariance matrix of a Gaussian model \textit{vis-\`a-vis} inference for the half precision, we also give the inverse of such matrices, easily computed via the Sherman-Morrison-Woodbury formula (sec. 0.7.4 of \cite{Horn-Johnson}) for the inverse of a rank-1 updated matrix.

\begin{lemma}\label{cs_inverse_entries}
    If $\textbf B = (b_1-b_2)\textbf I + b_2\boldsymbol 1\boldsymbol 1^\top$ is $d\times d$ compound symmetric, then $\textbf B^{-1} = (b_1'-b_2')\textbf I + b_2'\boldsymbol 1\boldsymbol 1^\top$ is compound symmetric as well, where
    \begin{align*}
        b_1' &= \frac{b_1+(d-2)b_2}{(b_1-b_2)(b_1+(d-1)b_2)}\\
        b_2' &= -\frac{b_2}{(b_1-b_2)(b_1+(d-1)b_2)}.
    \end{align*}
\end{lemma}

One observes that if the off-diagonal entry of compound symmetric $\textbf B$ is positive (negative), the off-diagonal of $\textbf B^{-1}$ is negative (positive).

\section{Details for Section \ref{sec:marg_cond}}\label{app:details}

\subsection{Notation and Terminology}

We frequently refer to the usual gamma function $\Gamma(\alpha) := \int_0^{\infty}x^{\alpha-1}\exp\{-x\}dx$ and beta function $\text{B}(\alpha,\beta) := \int_0^1 x^{\alpha-1}(1-x)^{\beta-1}dx = \frac{\Gamma(\alpha)\Gamma(\beta)}{\Gamma(\alpha+\beta)}$ for $\alpha,\beta>0$. Our analyses require knowledge of Kummer's confluent hypergeometric function \citep{Kummer-1837}, denoted by
\begin{equation*}
    {}_1F_1(\alpha,\beta,\lambda) := \sum_{n=0}^\infty \frac{(\alpha)_n}{(\beta)_n}\frac{\lambda^n}{n!}
\end{equation*}
for $\alpha,\beta>0$ and $\lambda \in \mathbbm R$, as well as Gauss' generalized hypergeometric function
\begin{equation*}
    {}_2F_1(\alpha,\beta,\gamma,\lambda) := \sum_{n=0}^\infty \frac{(\alpha)_n(\gamma)_n}{(\beta)_n}\frac{\lambda^n}{n!}
\end{equation*}
for $\alpha,\beta,\gamma> 0$ and $|\lambda|<1$. Here, $(\alpha)_n$ is the rising factorial or Pochhammer symbol, known to be equivalent to
\begin{equation*}
    (\alpha)_n = \frac{\Gamma(\alpha+n)}{\Gamma(\alpha)}.
\end{equation*}
While much has been written about both ${}_1F_1$ and ${}_2F_1$, we only require knowledge of a few of their properties, in particular that
\begin{equation}\label{Kummer_transform}
    {}_1F_1(\alpha,\alpha+\beta,\lambda) = \exp\{\lambda\}{}_1F_1(\beta,\alpha+\beta,-\lambda),
\end{equation}
a result known as Kummer's transformation (eq. 13.1.27 in \cite{AS-1964}), and that
\begin{equation}\label{spec_case}
    {}_2F_1(\alpha,\beta,\beta,\lambda) = \frac{1}{(1-\lambda)^\alpha}
\end{equation}
(eq. 15.1.8 in \cite{AS-1964}).

\subsection{Kummer-Beta Distribution}
The Kummer-beta distribution \cite{Ng-1995, Nagar-2002} generalizes the usual beta distribution to include an exponential term in the random variable which offers greater weight to one end of the unit interval to permit greater probability thereto \cite{Cordeiro-2014}. For $\alpha,\beta>0$ and $\lambda \in \mathbbm R$ one writes $X\sim KB(\alpha,\beta,\lambda)$ to indicate that the random variable $X$ has such a distribution, with density function
\begin{equation*}
    f_X(x;\alpha,\beta,\lambda) = \begin{cases}
    \frac{x^{\alpha-1}(1-x)^{\beta-1}\exp\{-\lambda x\}}{Z(\alpha,\beta,\lambda)}, \ \ x\in(0,1)\\
    0, \ \ \text{otherwise}
    \end{cases},
\end{equation*}
where the normalization constant is
\begin{equation}\label{Kummer_beta_Z}
    Z(\alpha,\beta,\lambda) = \text{B}(\alpha, \beta){}_1F_1(\alpha, \alpha+\beta,-\lambda).
\end{equation}
Note that when $\lambda=0$ then $X\sim \text{B}(\alpha, \beta)$ as a special case. We generalize this density by considering, for $a>0$ and $b\in \mathbbm R$, the random variable $Y = aX + b$, which we call a shifted/scaled Kummer-beta random variable (denoted as $Y\sim KB_{a,b}(\alpha,\beta,\gamma)$). The density function of $Y$, which can be easily gleaned from the usual change of variables, is the subject of the following lemma.
\begin{lemma}\label{noncentral_kb}
    If $X\sim KB(\alpha,\beta,\lambda)$, $a>0$, and $b\in \mathbbm R$, then the density of $Y=aX+b$ is
    \begin{equation*}
        f_Y(y;\alpha,\beta,\lambda,a,b) = \begin{cases}\frac{\exp\{\lambda\frac{b}{a}\}}{a^{\alpha+\beta-1}Z(\alpha,\beta,\lambda)}(y-b)^{\alpha-1}(a+b-y)^{\beta-1}\exp\{-\frac{\lambda}{a}y\}, \ \ y\in(b,a+b)\\
        0, \ \ \text{otherwise}
        \end{cases}
    \end{equation*}
\end{lemma}
\noindent\textbf{Proof: }
The linear relationship may be rearranged as $X=\frac{Y-b}{a}$, with Jacobian $\frac{dX}{dY} = \frac{1}{a}$. Moreover it's clear the said relationship is a bijection from $(0,1)$ to $(b,a+b)$, hence $Y$ has strictly positive support on the latter interval. Using change-of-variables we have
\begin{align*}
    f_Y(y;\alpha,\beta,\lambda,a,b) &= f_X\bigg(\frac{y-b}{a};\alpha,\beta,\lambda\bigg)\frac{1}{a}\\ &=\frac{1}{aZ(\alpha,\beta,\lambda)}\bigg(\frac{y-b}{a}\bigg)^{\alpha-1}\bigg(1-\frac{y-b}{a}\bigg)^{\beta-1}\exp\bigg\{-\lambda\bigg(\frac{y-b}{a}\bigg)\bigg\}\\
    &= \frac{\exp\{\lambda\frac{b}{a}\}}{a^{\alpha+\beta-1}Z(\alpha,\beta,\lambda)}(y-b)^{\alpha-1}(a+b-y)^{\beta-1}\exp\bigg\{-\frac{\lambda}{a}y\bigg\}
\end{align*}
on $(b,a+b)$, and 0 otherwise.\qed

That $a$ be strictly positive is hardly necessary to arrive at such a generalization. Indeed, we enforce such a restriction to avoid absolute values in the final density. If one is interested in generalizing for $a<0$, first note that Kummer's transformation (\ref{Kummer_transform}) gives $X':=1-X \sim KB(\beta,\alpha, -\gamma)$; when $a$ is negative we can see that $Y$ is a transformation of $X'$ as described above.

\subsection{Convolved-Gamma Distribution}

The marginal prior we derive for the diagonal entry of $\boldsymbol{\mathcal H}$ takes the form of a gamma density, albeit one generalized with Kummer's hypergeometric function. One writes $X\sim C\Gamma(\alpha,\beta,\lambda,\delta)$ to indicate that the random variable $X$ has a \textit{convolved-gamma} distribution with parameters $\alpha,\beta ,\lambda> 0$ and $-\infty<\delta<\beta$ and density
\begin{equation*}
    f_X(x;\alpha,\beta,\lambda,\delta) = \begin{cases}
        \frac{\beta^\alpha (\beta-\delta)^\lambda}{\Gamma(\alpha+\lambda)}x^{\alpha+\lambda-1}\exp\{-\beta x\}{}_1F_1(\lambda,\alpha+\lambda,\delta x), & \text{ if } x>0\\
        0, & \text{otherwise}
    \end{cases}.
\end{equation*}
\cite{DiSalvo-2006} obtained this density for the sum of independent $\Gamma(\alpha, \beta)$ and $\Gamma(\lambda,\beta-\delta)$ random variables, and has since been used to model the time-activity of a thyroidal imaging agent \citep{Wesolowski-2016}. This distribution includes the usual gamma density as a special case; for example, $C\Gamma(\alpha,\beta,\lambda,0) = \Gamma(\alpha+\lambda,\beta)$ and $\lim_{\alpha\to0} C\Gamma(\alpha,\beta,\lambda,\delta) = \Gamma(\lambda,\beta-\delta)$.

\subsection{Marginal and Conditional Priors of the Entries of the Half-Precision}

We observe that (\ref{eq:prec_prior_not_norm}) may be rewritten as
\begin{align*}
    \exp\{-\eta_1\beta_1\}\times \bigg((\eta_1-\eta_2)^{\frac{m(d-1)}{2}}(\eta_1 + (d-1)\eta_2)^{\frac{m}{2}}\exp\{- \eta_2\beta_2\}\bigg),
\end{align*}
where $\beta_1 := \text{tr}(\boldsymbol B)$ and $\beta_2 := \boldsymbol1^\top\boldsymbol B\boldsymbol1-\text{tr}(\boldsymbol B)$; Lemma \ref{trace_inequality} implies $-\beta_1<\beta_2<(d-1)\beta_1$. Further massaging of (\ref{eq:prec_prior_not_norm}) reveals that it may be written as
\begin{align*}
    \exp\{-\eta_1\beta_1\}\int_{-\frac{\eta_1}{d-1}}^{\eta_1} (\eta_1-\eta_2)^{\frac{m(d-1)}{2}}(\eta_1 + (d-1)\eta_2)^{\frac{m}{2}}\exp\{- \eta_2\beta_2\}\partial\eta_2\\
    \times \frac{(\eta_1-\eta_2)^{\frac{m(d-1)}{2}}(\eta_1 + (d-1)\eta_2)^{\frac{m}{2}}\exp\{- \eta_2\beta_2\}}{\int_{-\frac{\eta_1}{d-1}}^{\eta_1} (\eta_1-\eta_2)^{\frac{m(d-1)}{2}}(\eta_1 + (d-1)\eta_2)^{\frac{m}{2}}\exp\{- \eta_2\beta_2\}\partial\eta_2},
\end{align*}
rendering the non-normalized joint density of $\boldsymbol\eta$ as the product of a (non-normalized) marginal density of $\eta_1$ and a (normalized) density of $\eta_2$ given $\eta_1$. The latter looks like a non-central Kummer-Beta density on the interval $(-\frac{\eta_1}{d-1}, \eta_1)$ with a parameterization $(\alpha,\beta,\gamma)$ satisfying
\begin{align*}
    \alpha-1 &= \frac{m}{2}\\
    \beta-1 &=\frac{m(d-1)}{2}\\
    -\frac{\lambda(d-1)}{d\eta_1} &= \beta_{2};
\end{align*}
i.e., $\eta_2 \ | \ \eta_1$ is non-central Kummer-Beta with
\begin{align*}
    a &= \frac{d\eta_1}{d-1}\\
    b &= \frac{-\eta_1}{d-1}\\
    \alpha &= \frac{m+2}{2}\\
    \beta &= \frac{m(d-1)+2}{2}\\
    \lambda &= \frac{\eta_1d\beta_2}{d-1}.
\end{align*}

Thus, by Lemma \ref{noncentral_kb} the integral is precisely
\begin{equation*}
\frac{(\frac{d\eta_1}{d-1})^{\frac{md+4}{2}-1}\text{B}(\frac{m+2}{2},\frac{m(d-1)+2}{2}){}_1F_1(\frac{m+2}{2},\frac{md+4}{2}, -\frac{\eta_1d\beta_2}{d-1})}{\exp\{-\frac{\eta_1\beta_2}{d-1}\}}.
\end{equation*}
If we ignore the terms constant with respect to $\eta_1$ and multiply by $\exp\{-\beta_1\eta_1\}$, we have that the marginal density of $\eta_1$ is proportional to
\begin{equation*}
    \eta_1^{\frac{md+4}{2}-1}\exp\bigg\{-\big(\beta_1-\frac{\beta_2}{d-1}\big)\eta_1\bigg\}{}_1F_1\bigg(\frac{m+2}{2},\frac{md+4}{2}, -\frac{\eta_1d\beta_2}{d-1}\bigg),
\end{equation*}
i.e., a convolved-gamma density with parameters
\begin{align*}
    \alpha &= \frac{m(d-1)+2}{2}\\
    \lambda &= \frac{m+2}{2}\\
    \beta &= \beta_1-\frac{\beta_2}{d-1}\\
    \delta &= -\frac{d\beta_2}{d-1}
\end{align*}
and normalization constant\footnote{Note Lemma \ref{trace_inequality} implies that $\delta < \beta$, resulting in a valid parameterization.} (after some simplification)
\begin{equation*}
     \frac{\Gamma(\frac{md+4}{2})(d-1)^{\frac{m(d-1)+2}{2}}}{((d-1)\beta_1-\beta_2)^{\frac{m(d-1)+2}{2}}(\beta_1+\beta_2)^{\frac{m+2}{2}}}.
\end{equation*}

Taking all this into account, the joint prior for $\boldsymbol\eta$ is
\begin{equation}
    Z_d(m,\beta_1,\beta_2)(\eta_1-\eta_2)^{\frac{m(d-1)}{2}} (\eta_1 + (d-1)\eta_2)^{\frac{m}{2}}\exp\{-\beta_1\eta_1-\beta_2\eta_2\}
\end{equation}
where
\begin{equation}
    Z_d(m,\beta_1,\beta_2) := \frac{(\beta_1+\beta_2)^{\frac{m+2}{2}} ((d-1)\beta_1-\beta_2)^{\frac{m(d-1)+2}{2}}}{d^{\frac{md+2}{2}} \Gamma(\frac{m(d-1)+2}{2})\Gamma(\frac{m+2}{2})},
\end{equation}
with positive support on $\mathcal C_d$, as desired.

The prior expectation of $\boldsymbol\eta$ is the subject of the following theorem.

\begin{theorem}\label{thm:distribution}
    If
    \begin{equation*}
        \eta_1 \sim C\Gamma\bigg(\frac{m(d-1)+2}{2}, \frac{m+2}{2}, \beta_1-\frac{\beta_2}{d-1}, -\frac{d\beta_2}{d-1}\bigg)
    \end{equation*}
    and
    \begin{equation*}
        \eta_2\ | \ \eta_1 \sim KB_{\frac{d\eta_1}{d-1}, -\frac{\eta_1}{d-1}}\bigg(\frac{m+2}{2}, \frac{m(d-1)+2}{2}, \frac{\eta_1 d \beta_2}{d-1}\bigg)
    \end{equation*}
    then $\boldsymbol\eta$ is exactly equal in distribution to $\textbf{C}_d \boldsymbol Y$
    where $\textbf{C}_d = \begin{bmatrix}
        1 &1\\
        1 &-\frac{1}{d-1}
    \end{bmatrix},\ Y_1\sim\Gamma(\frac{m+2}{2}, \beta_1+\beta_2),\ Y_2\sim\Gamma(\frac{m(d-1)+2}{2}, \beta_1-\frac{\beta_2}{d-1})$, and $Y_1$ and $Y_2$ are independent.
\end{theorem}
\begin{proof}
    As $\eta_1\sim C\Gamma\big(\frac{m(d-1)+2}{2}, \frac{m+2}{2}, \beta_1-\frac{\beta_2}{d-1}, -\frac{d\beta_2}{d-1}\big)$, it may be written as the convolution of independent $\Gamma(\frac{m(d-1)+2}{2}, \beta_1-\frac{\beta_2}{d-1})$ and $\Gamma(\frac{m+2}{2}, \beta_1+\beta_2)$ random variables, since
    \begin{align*}
        \bigg(\beta_1-\frac{\beta_2}{d-1}\bigg)-\bigg(-\frac{d\beta_2}{d-1}\bigg) = \beta_1+\beta_2.
    \end{align*}
    Hence, $\eta_1$ equals in distribution the sum of the gamma variables given in the statement of the theorem.

    We shall compute the distribution of $\eta_2$ from its moment generating function. We have
    \begin{align*}
        \mathbb E&\lbrack e^{t\eta_2}\rbrack \\ &=\int_{\mathcal C_d} Z_d(m,\beta_1,\beta_2)(\eta_1-\eta_2)^{\frac{m(d-1)}{2}} (\eta_1 + (d-1)\eta_2)^{\frac{m}{2}}\exp\{-\beta_1\eta_1-\beta_2\eta_2\} \times\exp\{t\eta_2\}\partial\boldsymbol\eta\\
        &= Z_d(m,\beta_1,\beta_2) \int_{\mathcal C_d} (\eta_1-\eta_2)^{\frac{m(d-1)}{2}} (\eta_1 + (d-1)\eta_2)^{\frac{m}{2}}\exp\{-\beta_1\eta_1-(\beta_2-t)\eta_2\}\partial\boldsymbol\eta\\
        &= \frac{Z_d(m,\beta_1,\beta_2)}{Z_d(m,\beta_1,\beta_2-t)}
    \end{align*}
    provided that $t$ satisfies $-\beta_1 < \beta_2-t < (d-1)\beta_1$, i.e., $t\in \big(\beta_2-(d-1)\beta_1, \ \beta_2-\beta_1\big)$ an open interval in $\mathbb R$ containing 0. Thus
    \begin{align*}
        \mathbb E\lbrack e^{t\eta_2}\rbrack &= \frac{Z_d(m,\beta_1,\beta_2)}{Z_d(m,\beta_1,\beta_2-t)}\\
        &= \frac{(\beta_1+\beta_2)^{\frac{m+2}{2}} ((d-1)\beta_1-\beta_2)^{\frac{m(d-1)+2}{2}}}{(\beta_1+\beta_2-t)^{\frac{m+2}{2}} ((d-1)\beta_1-\beta_2+t)^{\frac{m(d-1)+2}{2}}}\\
        &=\bigg(\frac{\beta_1+\beta_2}{\beta_1+\beta_2-t}\bigg)^{\frac{m+2}{2}} \times \bigg(\frac{(d-1)\beta_1-\beta_2}{(d-1)\beta_1-\beta_2+t}\bigg)^{\frac{m(d-1)+2}{2}}.
    \end{align*}

    The first multiplicand in the final line is the moment generating function of $Y_1$ given in the theorem statement. The second multiplicand, after some rearrangement, is exactly
    \begin{equation*}
        \bigg(\frac{\beta_1-\frac{\beta_2}{d-1}}{\beta_1-\frac{\beta_2}{d-1}+\frac{t}{d-1}}\bigg)^{\frac{m(d-1)+2}{2}},
    \end{equation*}
    i.e., the moment generating function of $Y_2$ evaluated at $-\frac{t}{d-1}$ or, equivalently, the moment generating function of $-\frac{Y_2}{d-1}$.

    We now have that $\eta_1$ and $\eta_2$ are respectively equal in distribution to $Y_1+Y_2$ and $Y_1-\frac{Y_2}{d-1}$ which, after being expressed in matrix form, completes the result.
\end{proof}

\section{EM Algorithm for the Marginal Formulation of the Random Intercept Model}\label{app:EM_for_RI}

We suggested above that one may potentially initialize the Gibbs sampler for the test of positivity of a common within-class correlation at the terminal values of an expecation-maximization (EM) algorithm. While such methods for the conditional formulation of the random-intercept model have been thoroughly discussed in the literature (e.g, \cite{Kato-2004, Ippel-2009, Suqin-2014}), we could not find any for the marginal formulation (\ref{marginal_model}). For the sake of thoroughness we detail here how such a method may be derived.

Suppose we are in the marginal characterization of the random-intercept model, i.e., with vectors $\boldsymbol X_{\cdot j}, \ j = 1,\dots, J$ of varying lengths $d_j$ (with sum $D$ and maximum length $d_\text{max}$) arranged into
\begin{equation*}
    \boldsymbol X = \begin{bmatrix}
        \boldsymbol X_{\cdot 1}\\
        \vdots\\
        \boldsymbol X_{\cdot J}
    \end{bmatrix} \sim \mathcal N(\mu\boldsymbol 1_D, \boldsymbol\Sigma^*)
\end{equation*}
in which
\begin{equation*}
    \boldsymbol\Sigma^* = \bigoplus_{j=1}^J \big((\sigma_1-\sigma_2) \textbf I_{d_j} + \sigma_2\boldsymbol 1_{d_j}\boldsymbol 1_{d_j}^\top\big)
\end{equation*}
with $\mu \in \mathbbm R$ and $\boldsymbol\sigma \in \mathcal C_{d_\text{max}}$. We shall take a frequentist point of view and assume that no prior has been placed on the parameters.

Any EM algorithm begins with a complete-data extension of the model which, when completely observed, would render usual maximum likelihood estimation straightforward. Following initialization of the parameter estimate, the algorithm alternates between an E-step, which computes the expectation under the current parameter of the complete-data log-likelihood conditional on the observed data, and an M-step, in which the parameters are optimized by maximizing the resulting expression. For the particular case in which the complete-data model is a full-rank exponential family, the E-step only involves computing conditional expectations of the complete-data sufficient statistics; and the M-step performs maximum likelihood estimation supposing the said conditional expectation had been observed (\cite{DLR}, Section 2).

Since maximum likelihood estimation of $(\mu, \boldsymbol\sigma)$ can be performed when the group sizes are all equal, we propose extending each group observation in the marginal model with additional $\boldsymbol Y_{\cdot j}$ of length $d_{\text{max}}-d_j$ such that
\begin{equation}\label{complete_marginal_model}
    \begin{bmatrix}
        \boldsymbol X_{\cdot j}\\
        \boldsymbol Y_{\cdot j}
    \end{bmatrix}
    \overset{i.i.d.}{\sim} \mathcal N(\mu\boldsymbol1_{d_{\text{max}}}, \boldsymbol\Sigma), \ j = 1,\dots, J
\end{equation}
where $\boldsymbol\Sigma = (\sigma_1-\sigma_2)\textbf I_{d_\text{max}} + \sigma_2\boldsymbol 1_{d_\text{max}}\boldsymbol 1_{d_\text{max}}^\top$. This is a fully rank-3 exponential family, with sufficient statistics
\begin{align*}
    T_\mu(\boldsymbol X, \boldsymbol Y) &=\sum_{j=1}^J \bigg(\sum_{i=1}^{d_j} X_{ij} + \sum_{i=d_j+1}^{d_{\text{max}}} Y_{ij}\bigg)\\
    \boldsymbol T_{\boldsymbol\sigma}(\boldsymbol X, \boldsymbol Y) &=\sum_{j=1}^J\begin{bmatrix}
        \boldsymbol X_{\cdot j} \boldsymbol X_{\cdot j}^\top &\boldsymbol X_{\cdot j} \boldsymbol Y_{\cdot j}^\top\\
        \boldsymbol Y_{\cdot j} \boldsymbol X_{\cdot j}^\top &\boldsymbol Y_{\cdot j} \boldsymbol Y_{\cdot j}^\top
    \end{bmatrix}.
\end{align*}
Let $\boldsymbol\theta^* = (\mu^*,\boldsymbol\sigma^*)$ be the most recent iterate of the parameters; since each $(\boldsymbol X_{\cdot j}, \boldsymbol Y_{\cdot j})$ is jointly normally distributed under $\boldsymbol\theta^*$, conditional on $\boldsymbol X = \boldsymbol x$, the $\boldsymbol Y_{\cdot j}$ are independent of each other and respectively normally distributed in $\mathbbm R^{d_{\text{max}}-d_j}$ (\cite{Bic-Dok}, Theorem B.6.5); we eventually obtain
\begin{equation}\label{conditional_dist}
    \boldsymbol Y_{\cdot j} \ | \ \boldsymbol X_{\cdot j} = \boldsymbol x_{\cdot j} \sim \mathcal N(\mu^*_j\boldsymbol1_{d_\text{max}-d_j}, \boldsymbol\Sigma^*_j)
\end{equation}
where
\begin{align*}
    \mu_{j}^* &= \frac{\mu^*(\sigma_1^*-\sigma_2^*)+\sigma_2^*\sum_{i=1}^{d_j}x_{ij}}{\sigma_1^*+(d_j-1)\sigma_2^*}\\
    \boldsymbol\Sigma_{j}^* &= (\sigma_1^*-\sigma_2^*)\textbf I_{d_{\text{max}}-d_j}+\frac{\sigma_2^*\boldsymbol1_{d_\text{max}-d_j}\boldsymbol1^\top_{d_\text{max}-d_j}}{\sigma_1^*+(d_j-1)\sigma_2^*}.
\end{align*}
To compute
\begin{align}
    T_{\mu}^* &= \mathbbm E_{\boldsymbol\theta^*}\lbrack T_{\mu}(\boldsymbol X,\boldsymbol Y)\ | \ \boldsymbol X=\boldsymbol x\rbrack\\
    \boldsymbol T_{\boldsymbol\sigma}^* &= \mathbbm E_{\boldsymbol\theta^*}\lbrack\boldsymbol T_{\boldsymbol\sigma}(\boldsymbol X,\boldsymbol Y)\ | \ \boldsymbol X=\boldsymbol x\rbrack
\end{align}
for the E-step, we find that we must only compute at most three conditional expectations for each group of size smaller than $d_{\text{max}}$, those of $Y_{ij}$, $Y_{ij}^2$, and $Y_{ij}Y_{kj}$ with $i\neq k$. Each of these terms may be easily gleaned from $\mu_j^*$ and $\boldsymbol\Sigma_j^*$ above:
\begin{align*}
    Y_j^* &:= \mathbbm E_{\boldsymbol\theta^*}\lbrack Y_{ij} \ | \ \boldsymbol X_{\cdot j} = \boldsymbol x_{\cdot j}\rbrack = \mu_j^*\\
    (Y_j^2)^* &:= \mathbbm E_{\boldsymbol\theta^*}\lbrack Y_{ij}^2 \ | \ \boldsymbol X_{\cdot j} = \boldsymbol x_{\cdot j}\rbrack = \frac{(\sigma_1^*-\sigma_2^*)^2 + (d_j\sigma_1^*+1)\sigma_2^*}{\sigma_1^*+(d_j-1)\sigma_2^*} + (\mu_j^*)^2\\
    (Y_jY_j')^* &:=\mathbbm E_{\boldsymbol\theta^*}\lbrack Y_{ij}Y_{kj} \ | \ \boldsymbol X_{\cdot j} =\boldsymbol x_{\cdot j}\rbrack=  \frac{\sigma_2^*}{\sigma_1^*+(d_j-1)\sigma_2^*} + (\mu_j^*)^2.
\end{align*}

The M-step only consists of performing maximum likelihood estimation having observed the sufficient statistics $T_{\mu}^*$ and $\boldsymbol T_{\boldsymbol\sigma}^*$ to update the parameter estimates. In doing so, we eventually obtain $\boldsymbol{\hat\theta} = (\hat\mu, \boldsymbol{\hat\sigma})$ where
\begin{align*}
    \hat\mu &= \frac{T^*_{\mu}}{J d_{\text{max}}}\\
    \hat\sigma_1 &= \frac{\text{tr}(\boldsymbol T_{\boldsymbol\sigma}^*)-Jd_{\text{max}}\hat\mu^2}{d_{\text{max}}(J-1)}\\
    \hat\sigma_2 &= \frac{\boldsymbol1^\top_{d_\text{max}}\boldsymbol T_{\boldsymbol\sigma}^*\boldsymbol1_{d_\text{max}}-\text{tr}(\boldsymbol T_{\boldsymbol\sigma}^*)-d_\text{max}(d_\text{max}-1)J\hat\mu^2}{d_\text{max}(d_\text{max}-1)(J-1)}.
\end{align*}
We alternate between the E- and M-steps until some convergence criterion has been achieved, e.g., the increases to the log-likelihood fall beneath a specified threshold. We propose a simple initialization for this procedure. Simply take $\mu_0 = \bar x, \sigma_{20} = 0$, and $\sigma_{10}$ as the average sample variance of all the $\boldsymbol x_{\cdot j}$ of length $d_{\text{max}}$.

\end{document}